\documentclass[11pt]{amsart}

\usepackage[a4paper,margin=25mm]{geometry}
\usepackage{color}

\usepackage{verbatim}
\usepackage{enumerate}
\usepackage[breaklinks]{hyperref}

\newtheorem{theorem}{Theorem}
\newtheorem{lemma}{Lemma}
\newtheorem{corollary}{Corollary}

\theoremstyle{definition}
\newtheorem{definition}{Definition}
\theoremstyle{remark}

\newtheorem{clm}{Claim}

\newcommand{\PSPACE}{{\sf PSPACE}}
\newcommand{\NPclass}{{\sf NP}}
\newcommand{\Ocomp}{\ensuremath{\mathcal{O}}}
\newcommand{\poly}{\text{\sf poly}}

\newcommand{\sol}[1]{\ensuremath{S(#1)}}
\newcommand{\solution}{\ensuremath{S}}

\usepackage[ruled]{algorithm}
\usepackage[noend]{algpseudocode}
\newcommand{\algofont}[1]{\textnormal{\selectfont\sffamily#1}}

\newcommand{\CPref}[1]{(\hyperref[cr #1]{CP#1})}
\newcommand{\CCref}[1]{(\hyperref[cc #1]{CC#1})}
\newcommand{\CFLref}[1]{(\hyperref[cfl #1]{CFL#1})}
\newcommand{\rownanie}[1][]{\ensuremath{u_{#1} = v_{#1}}}
\newcommand{\rownaniep}[1][]{{\textnormal{`\rownanie[#1]'}}}

\newcommand{\algtreechildcomp}{\algofont{TreeLeafComp}}
\newcommand{\algchildcompncr}{\algofont{LeafComp}}
\newcommand{\algtreecomp}{\algofont{TreeComp}}
\newcommand{\algtreepaircomp}{\algofont{TreePartitionComp}}
\newcommand{\algpaircompncr}{\algofont{PartitionComp}}
\newcommand{\algtreechaincomp}{\algofont{TreeChainComp}}
\newcommand{\algchaincompncr}{\algofont{ChainComp}}
\newcommand{\algpop}{\algofont{Pop}}
\newcommand{\alggenpop}{\algofont{GenPop}}
\newcommand{\algprefsuff}{\algofont{CutPrefSuff}}
\newcommand{\algsolveeq}{\algofont{ContextEqSat}}

\newcommand{\twodots}{\mathinner{\ldotp\ldotp}}
\DeclareMathOperator{\ar}{{ar}}
\newcommand{\variables}{\ensuremath{\mathcal X }}
\newcommand{\contextvar}{\ensuremath{\mathcal V }}

\title{Context unification is in PSPACE}

\author[A.\ Je\.z]{Artur Je\.z}
\thanks{Research supported by Humboldt Foundation Research Fellowship for Postdoctoral Researchers, 2013--14}
\address{
Max Planck Institute f\"ur Informatik,\\
Campus E1 4,  DE-66123 Saarbr\"ucken, Germany\\
\and
Institute of Computer Science, University of Wroc{\l}aw \\
ul.\ Joliot-Curie~15, 50-383 Wroc{\l}aw, Poland\\
\texttt{aje@cs.uni.wroc.pl}}

\begin{document}

\begin{abstract}
Contexts are terms with one `hole', i.e.\ a place in which we can substitute an argument.
In context unification we are given an equation over terms with variables representing contexts
and ask about the satisfiability of this equation.
Context unification is a natural subvariant of second-order unification, which is undecidable,
and a generalization of word equations, which are decidable, at the same time.
It is the unique problem between those two whose decidability is uncertain
(for already almost two decades).
In this paper we show that the context unification is in \PSPACE.
The result holds under a (usual) assumption that the first-order signature is finite.

This result is obtained by an extension of the recompression technique, recently developed by the author and used in particular
to obtain a new \PSPACE{} algorithm for satisfiability of word equations, to context unification.
The recompression is based on performing simple compression rules (replacing pairs of neighbouring function symbols),
which are (conceptually) applied on the solution of the context equation and modifying the equation in a way so that such
compression steps can be in fact performed directly on the equation, without the knowledge of the actual solution.
\end{abstract}
\maketitle

\section{Introduction}
\subsection{Context unification}

Solving equations, whether they are over groups, fields, semigroups, terms or any other objects,
was always a central point in mathematics and the corresponding decision problems
received a lot of attention in the theoretical computer science community.
Solving equations can be equally seen as unification problem, as we are to unify two objects (with some variables).

Context unification is one of prominent problems of this kind, let us first introduce the objects we will work on.
A ground context is a ground term with exactly one occurrence of a special constant that
represents a missing argument.
Ground contexts can be applied to ground terms, which results in a replacement of
the special constant by the given ground term;
similarly we can define a~composition of two ground contexts, which is again a ground context.
Hence we can built terms using ground contexts, treating them as function symbols of arity $1$.

In context unification we are given a finite signature,
a set of variables (which shall denote ground terms) and a set of so-called context variables (which shall denote ground contexts).
Using those variables we can built terms: we simply treat each context variable as a function symbol of arity one
and each variable as a constant.
A context equation is an equation between two such terms
and a solution of a context equation assigns to each context variable a ground context (over the given input signature)
and to each variable a ground term (over the same signature)
such that both sides of the equation evaluate to the same (ground) term.
The context unification is the decision problem, whether a context equation has a solution
(as in some sense we unify the two contexts on the sides of the equation).

Context unification was introduced by Comon~\cite{Comon98,Comon98a} (who also coined the name)
and independently by Schmidt-Schau\ss~\cite{definition2}.
It found usage in analysis of rewrite systems with membership constraints~\cite{Comon98,Comon98a},
analysis of natural language~\cite{NiehrenPR97Cade,NiehrenPR97},
distributive unification~\cite{Schmidt-Schauss98}, bi-rewriting systems~\cite{birewrite}.

In a broader sense, context unification is a special case of second-order unification,
in which the argument of the second-order variable $X$ can be used unbounded number of times in the substitution term for $X$
(also, there may be many parameters for a second order variable, this is however not an essential difference).
On the other hand, when the underlying signature is restricted to the case when only unary function symbols and constants are allowed,
the context equation is in fact a word equation
(in this well-known problem we are given an equation $u = v$, where $u$ and $v$ are strings of letters and variables and we are
to substitute the variables with strings so that this formal equation is turned into a true equality of strings).
The second order unification is known to be undecidable~\cite{Goldfarb81},
(even in very restricted cases~\cite{Farmer91,Levy96,LevyV00})
however, the proofs do not apply to the case of context unification as they essentially use the fact that the argument may be used many times
in the substitution term.
On the other hand, the satisfiability of word equations is known to be decidable (in \PSPACE~\cite{PlandowskiSTOC})
and up to recently there were essentially only three different algorithms for this problem~\cite{Makanin,PlandowskiICALP,PlandowskiSTOC};
whether these algorithms generalise to context unification remains an open question.
Hence context unification is both upper and lower-bounded by two well-studied problems.

The problem gained considerable attention in the term rewriting community~\cite{RTAproblem90}, mainly for two reasons:
on one hand it is the only known natural problem which is subsumed by second order unification (which is undecidable)
and subsumes word equations (which are decidable) and on the other hand it has several ties to other problems, see Section~
\ref{subsec: connections}.

There was a large body of work focused on context unification and several partial results were obtained:
\begin{itemize}
	\item a fragment in which any occurrence of the same context variable is always applied to the same term is decidable~\cite{Comon98a};
	\item stratified context unification, in which for any occurrence of a fixed second-order variable $X$
the string of second-order variables from this occurrence to the root of the containing term is the same is decidable~\cite{stratifiedcontext}
(this problem is even known to be \NPclass-complete~\cite{stratifiedcontextcomplete} and in fact the result holds even for infinite signatures);
	\item a fragment in which every variable and context variable occurs at most twice
	(such equations are usually called \emph{quadratic}) is decidable~\cite{Levy96};
	\item a fragment when there are only two context variables is decidable~\cite{contexttwovar};
	\item the notion of exponent of periodicity, which is crucial in algorithms for solving word equations,
	is generalised to context unification and so is the exponential bound on it~\cite{contextexponent};
	\item context unification reduces to the fragment in which the signature contains only one binary symbol
	and constants~\cite{reduction};
	\item context unification with one context variable is known to be in \NPclass~\cite{onecontextvariable}.
\end{itemize}
Note that in most cases the corresponding variants of the general second order unification remain undecidable,
which gave hope that context unification is indeed decidable.

In this paper we show that context unification can be nondeterministically decided in space polynomial in $k$ and $n$,
where $n$ is the size of the context equation and $k$ is the maximal arity of function symbols in the signature
(this means that we can consider infinite signatures as long as the maximal arity is bounded;
however, this case in general reduces to the case of finite signatures).

\subsection{Extensions and connections to other problems}
\label{subsec: connections}
The context unification was shown to be equivalent to `equality up to constraint' problem~\cite{NiehrenPR97Cade}
(which is a common generalisation of equality constraints, subtree constraints and one-step rewriting constraints).
In fact one-step rewriting constraints, which is a problem extensively studied on its own,
are equivalent to stratified context unification~\cite{NiehrenPR97}.
It is known that the existential theory of one-step rewriting constraints is undecidable~\cite{Reinen98,Marcinkowski97,Vorobyov97}.
The case of general context unification was improved by Vorobyov, who showed 
that its $\forall$ $\exists^8$-equational theory is $\Pi_1^0$-hard~\cite{Vorobyov98}.

Some fragments of second order unification are known to reduce to context unification:
the \emph{bounded second order unification} we assume that 
the number of appearances of the argument of the second-order variable
in the substitution term is bounded by a constant
(note that it \emph{can be zero} and this is the crucial difference with context unification).
This fragment on one hand easily reduces to context unification and on the other hand it is known to be decidable~\cite{boundedsecondorder}
(in fact its generalisation to higher-order unification is decidable as well~\cite{Schmidt-SchaussS05}
and it is known that bounded second order unification is \NPclass-complete~\cite{stratifiedcontextcomplete}).
In particular, the work presented here imply the results on bounded second order unification.

The context unification can be also extended by allowing some additional constraints on the (context) variables,
a natural one allows the usage of the tree-regular constraints
(i.e.\ we assume that the substitution for the (context) variable comes from a certain regular set of trees).
It is known that such an extension is equivalent to the linear second order unification~\cite{LevyVill00},
defined by Levy~\cite{Levy96}:
in essence, the linear second order unification allows bounding variables on different levels of the function,
which makes direct translations to context unification infeasible,
however, usage of regular constraints gives enough additional power to actually encode such more complicated bounding.

Notice the usage of regular constraints is very popular in case of word equations,
in particular it is used in generalisations of the algorithm for word equation to the group case
and both Makanin's and Plandowski's algorithms can be generalised to word equations with regular constraints~\cite{Schulz90,Diekertfreegroups}.

\subsection{Recompression}
The connection between compression and word equations was first observed and used by Plandowski and Rytter~\cite{PlandowskiICALP},
who showed that each length-minimal solution (of size $N$) of the word equation (of size $n$) has $\poly(n,\log N)$ description (in terms of LZ77).
This connection was exploited more efficiently by Plandowski,
whose \PSPACE{} algorithm works on compressed representation of the word equation
(and uses some finely tuned word factorisations to process this equation).

The recompression method, introduced recently by the author, is based solely of compression:
it performs two simple compression steps
(replacing pairs of letters $ab$ by a new letter $c$, replacing maximal blocks $a^\ell$ by a new letter $a_\ell$)
on a word represented in some implicit way, say as a grammar, compression scheme or even as a solution of a word equation.
In order to make such compression steps applicable, the implicit representation is modified a bit,
for instance in case of word equations a variable $x$ is replaced with $ax$ or $xb$ where $a, b$ are letters.
The intuition behind such a modification is apparent: when we want to replace each $ab$ by $c$ then some of those substrings
appear explicitly in the equation (which are hence easy to replace), some in the substitution for the variables
(which are thus replaced `implicitly' by changing the solution) and some are `crossing' between variables and letters in the equation
(or two variables).
The last case is the only problematic one, and it occurs (for a pair $ab$) when $ax$ occurs in the equation
and the solution for $x$ begins with $b$ (there is also the symmetric case).
In such a case replacing $ab$ is impossible.
To fix this problem, we modify the equation, by replacing $x$ with $bx$, thus `left-popping' the letter out of the variable.
It is easy to show that after left-popping $b$ and `right-popping' $a$ the pair $ab$ no longer has the problematic crossing
occurrences and so it can be replaced with $c$ in the equation.
The crucial observation is that when the compression are done in a proper way,
we can bound the size of the equation, as the number of letters popped into the equation is linear in the number of variables,
while the compression steps guarantee shortening of the equation by a constant factor.
Those two effects cancel each other out and so the equation has linear size.

This method turned out to be applicable to several problems
for implicit representations of words~\cite{fullyNFA,FCPM,wordequations,grammar,onevarlinear}
in particular it is applicable to the word equations problems in which case it yields a much simpler \PSPACE{}
algorithm that checks the satisfiability of the word equation (and returns a finite representation of all solutions)~\cite{wordequations}.
Retracing the compression steps yields an SLP (so a context free grammar generating a unique string)
for the size-minimal solution of the word equation and a simple analysis show that the size of this SLP is $\poly(n,\log N)$,
which yields an alternative proof of the result by Plandowski and Rytter~\cite{PlandowskiICALP} (with slightly better bounds).
Quite surprisingly, this algorithm, when restricted to the case of word equations with only one variable
yields a linear time algorithm~\cite{onevarlinear}, improving the previously known algorithms
(of course some further analysis and usage of tailored data structures is needed).

Applications of compression to fragments of context unification are known~\cite{Schmidt-SchaussS05}
and this paper extends the recompression method to terms in full generality.
In this way solving word equations using recompression~\cite{wordequations} generalises to solving context unification,
which in some sense fulfils the plan of extending the algorithms for word equations to context unification.
In particular, it provides a tree-grammar of size $\poly(n,\log N)$ generating a solution of a context equation.

A word can be seen as a term over signature containing only unary symbols (plus some constant at the bottom) and vice versa.
Thus the two compression operations for word equations generalise naturally to subterms containing only unary function symbols.
Hence the recompression for terms uses the two already mentioned operations
(which are applicable only to function symbols of arity one
\footnote{Note that by work of Levy~\cite{reduction} it is enough to consider context unification with constants and a single binary symbol.
However, our algorithm will transforms the input instance and it can introduce unary symbols.
So even if the input satisfies such a condition,
we cannot guarantee that the current context equation stored by the algorithm also satisfies it.})
but it also introduces another local compression rule, designed specifically for terms:
we replace a term $f(t_1,\ldots,t_{i-1},c,t_{i+1},\ldots,t_m)$ (where $c$ is a constant)
with $f'(t_1,\ldots,t_{i-1},t_{i+1},\ldots,t_m)$, where $f'$ is a fresh function symbol
(i.e.\ not used the context equation, it can however be in $\Sigma$).
While such a compression introduces new function symbols, it does not increase the maximal arity of functions in the signature,
which proves to be important (as the space consumption depends on this maximal arity).
This new rule requires also a generalisation of the variable replacements ($x$ by $ax$ or $xb$):
when $X$ denotes a context, we sometimes replace it with $a(X)$, where $a$ is a unary letter,
or $X(f(x_1,x_2,\ldots,x_{i-1},\Omega,x_i,\ldots,x_m))$, where $x_1,x_2,\ldots,x_m$ are new variables denoting full terms
and $\Omega$ denotes the place in which we apply the argument.

As in the case of word equations, 
the key observation is that while the variable replacements increase the size of the context equation
(proportionally to the number of occurrences of variables in the context equation),
the replacement rules guarantee that the size of the context equation is decreased by a constant factor
(for proper nondeterministic choices).
Those two effects cancel each out and the size of the context equation remains linear.

Note that the generalisation of the recompression to terms is independently considered also
in the joint work of the author and with Lohrey~\cite{treegrammar} on tree grammars
and the presentation of the recompression for terms there is similar to the one given here.

\subsection{Outline}
First, in Section~\ref{sec:trees}, we explain in detail how to generalise the recompression from strings to trees.
Then, in Section~\ref{sec:context unification}, we define formally the context unification and state some of its basic properties.
As a next step we identify the easy cases (so called \emph{non-crossing}),
in which the compression rules can be applied directly to the context equations,
see Section~\ref{sec: noncrossing compression}.
The main technical part of this paper is the explanation how to modify the context equations
so that each context equation is reduced to such an easy case mentioned above,
so that the compression schema can be applied directly to context equations, which is done in Section~\ref{sec:uncrossing};
this technique is called \emph{uncrossing} (and, despite its significance, is very easy).
The last Section~\ref{sec: main algorithm} wraps everything up, by presenting the full statement of the algorithm for context unification
as well as a (relatively simple) analysis of it as well as the polynomial bound on the space consumption.

\section{Compression of trees}
	\label{sec:trees}
In this section we generalise the technique of local compression of trees.
It is independently used in work of Je\.z and Lohrey on tree-grammar compressions
and the presentation there is similar~\cite{treegrammar}.

	\subsection{Labelled trees and their compression}
	\label{subsec: compressing trees}
We deal with rooted, ordered trees, usually denoted with letters $t$ or $s$.
Nodes are labelled with elements from a ranked alphabet $\Gamma$, i.e.\ each letter $a \in \Gamma$ has a fixed arity $\ar(f)$.
A tree is \emph{well-formed} if a node labelled with $f$ has exactly $\ar(f)$ children.\
Unless explicitly written, we consider only well-formed trees,
which can be equivalently seen as \emph{ground terms} over $\Gamma$.

In the following we usually consider a set of labels $\Sigma$ which is finite
but growing during the run of our algorithm.
We call the labels from $\Sigma$ \emph{letters} and
pay particular attention to letters of arity~$1$ (\emph{unary letters})
and to letters of arity $0$ (\emph{constants}).
On the other hand, $\Gamma$ (perhaps with some subscripts) is used for some subalphabet of $\Sigma$,
say, letters used in some particular tree or letters of arity at least $1$.

We want to replace fragments of a tree with new letters. Those fragments are not necessarily well-formed.
Thus we define a \emph{subtree} in a natural way, in general not necessarily well-formed,
but in such a case we explicitly mention it.
A \emph{pattern} is a tree (perhaps not well-formed) in which a node labelled with $f$ has \emph{at most} $\ar(f)$ children;
since we imagine a pattern as a part of a term with some of the subterm removed,
the $0 \leq m \leq \ar(f)$ children of $f$ in the pattern are numbered $1 \leq i_1 < i_2 < \dots < i_m \leq \ar(f)$
to denote which children of $f$ are those in a `real term'.
A subpattern of a tree $t$ is any subtree which is a pattern;
we often consider individual \emph{occurrences} of subpatterns of a tree $t$.
In this terminology, our algorithm will replace occurrences of subpatterns of $t$ in $t$
(the subtree rooted in children which are omitted in the subpattern are attached in the same order,
details are given later).

A \emph{chain} is a pattern that consists only of unary letters.
We consider $2$-chains, so consisting only of two unary letters (usually different)
and $a$-chains, which consists solely of letters $a$.
We treat chains as strings and write them in the string notation and `concatenate' them,
i.e.\ for two chains $s$ and $s'$ the $ss'$ denotes the chain obtained by attaching the top-most
node in $s$ to the bottom node in $s$.
A chain $t'$ that is a subpattern of $t$ is a \emph{chain subpattern} of $t$,
an occurrence of a chain subpattern $a^\ell$ is \emph{$a$-maximal} if it cannot be extended by $a$ nor up nor down.

\subsection{Local compression of trees}
We perform three types of compressions on a tree $t$, all of them replace subpatterns by a single letter:
\begin{description}
	\item[$a$-maximal chain compression]
	For a unary letter $a$ we replace each $a$-maximal chain subpattern $a^\ell$ by a new unary letter $a_\ell$
	(making the father of $a^\ell$ the father of $a_\ell$ and the unique child of $a^\ell$ the unique child of $a_\ell$).

	\item[$a,b$ pair compression]
	For two unary letters $a$ and $b$ we replace each $2$-chain subpattern $ab$ with a new unary letter $c$.

	\item[$(f,i,c)$ leaf compression] For a constant $c$ and letter $f$ of arity $\ar(f) = m \geq i \geq 1$,
	we replace each subtree $f(t_1,t_2,\ldots,t_{i-1},c,t_{i+1},\ldots,t_m)$ with
	$f'(t_1,t_2,\ldots,t_{i-1},t_{i+1},\ldots,t_m)$ where $f'$ is a fresh letter of arity $m-1$ added to $\Sigma$
	(intuitively: the constant $c$ is `absorbed' by its father).
\end{description}

Those operations are applied (in some specific order) on a tree $t$ until it is reduced to a single leaf.

Observe that the $a$-maximal chain compression and $a,b$ pair compression are direct translations
of the operations used in the recompression-based algorithm for word equations~\cite{wordequations}.
To be more precise, both those compressions affect only chains, return chains as well, and when a chain is treated as a string
the result of those compressions corresponds to the result of the corresponding operation on strings.
On the other hand, the leaf compression is a new operation that is designed specifically to deal with trees.

\subsubsection{Parallel compressions}
To make the compression more effective, we apply several compression steps in parallel: consider the $a$-maximal chain compression.
As $a$-maximal and $b$-maximal chain subpatterns do not overlap (it does not matter whether $a = b$ or not),
we can perform $a$-maximal chain compression for all $a \in \Gamma_1$ in parallel
(as long as the letters that are used to replace the chains are not taken from $\Gamma_1$).
We call the resulting procedure \algtreechaincomp$(\Gamma_1,t)$ or simply \emph{chain compression},
when $\Gamma_1$ and $t$ are clear from the context.

\begin{algorithm}[H]
	\caption{\algtreechaincomp$(\Gamma_1,t)$: Compression of chains of letters from $\Gamma_1$ in a tree $t$}
	\label{alg:treechaincomp}
	\begin{algorithmic}[1]
	\Require $\Gamma_1$ contains only unary letters
    \For{each $a \in \Gamma_1$} \Comment{Chain compression}
			\For{each $\ell \in \mathbb N$}
    		\State replace each $a$-maximal occurrence of chain subpattern $a^\ell$ in $t$ by $a_\ell$
			\EndFor
    \EndFor
	\end{algorithmic}
\end{algorithm}
An important property of the chain compression is that afterwards a father and son cannot be labelled with the same unary letter.
\begin{lemma}
	\label{lem: father and son}
After the chain compression performed for all unary letters in the obtained tree
there is no node labelled with the same unary letter as its father.
\end{lemma}
\begin{proof}
Suppose that in $t'$ there is a father and son labelled with the same unary letter $a$.
If $a$ was not introduced by the chain compression, then we arrive at a contradiction, as
those two letters should have been replaced with one unary letter.
If $a$ replaced some chain $b^\ell$ then we again obtain a contradiction,
as $aa$ represents a chain $b^{2\ell}$, so none of the two replaced $b^\ell$ was maximal.
\qedhere.
\end{proof}

We would like to perform also several $(f,i,a)$ compressions, for $f \in \Gamma_{\geq 1}$ and $a \in \Gamma_0$ in parallel.
Clearly for a fixed node labelled with $f \in \Gamma_{\geq 1}$ we can perform several different leaf compressions with its children
(that are from $\Gamma_0$) at the same time,
in this way we could define $(f,i_1,a_1,i_2,a_2,\ldots,i_\ell,a_\ell)$ leaf compression,
which replaces a subpattern $f$ with children $a_j$ at position $i_j$ for $j=1 \ldots, \ell$ with a new letter $f'$ with those children removed.
However, in this way two different $(f,i_1,a_1,i_2,a_2,\ldots,i_\ell,a_\ell)$ and
$(f,i_1',a_1',i_2',a_2',\ldots,i_{\ell'}',a_{\ell'}')$ leaf compressions could be applied to the same node labelled with $f$
and the result depends on the order of those two compressions, in particular they cannot be applied in parallel.
To remedy this, we apply the $(f,i_1,a_1,i_2,a_2,\ldots,i_\ell,a_\ell)$ leaf compression only to nodes that do not have children
labelled with letters from $\Gamma_0$ on positions other than $\{i_1,\ldots,i_\ell\}$.
More formally, the $(f,i_1,a_1,i_2,a_2,\ldots,i_\ell,a_\ell)$ leaf compression replaces each subtree
$f(t_1,t_2,\ldots, t_{i_1-1},a_1,t_{i_1+1},\ldots, t_{i_\ell-1},a_\ell,t_{i_\ell+1},\ldots,t_k)$ with $f'(t_1,\ldots,t_k)$
when $t_i \notin \Gamma_0$ for each $i \notin \{i_1,\ldots,i_\ell\}$.
Clearly such a compression can be also performed for different labels $f$ and different tuples $(i_1,a_1,i_2,a_2,\ldots,i_\ell,a_\ell)$
in parallel, as long as we do not try to compress also the letters introduced during the compression.
This is formalised in the following algorithm \algtreechildcomp$(\Gamma_{\geq 1},\Gamma_0,t)$,
when $\Gamma_{\geq 1}$, $\Gamma_0$, and $t$ are clear form the context, we simply talk about \emph{leaf compression}.

\begin{algorithm}[H]
	\caption{\algtreechildcomp$(\Gamma_{\geq 1},\Gamma_0,t)$: leaf compression for parent nodes in $\Gamma_{\geq 1}$,
	and leaf-children in $\Gamma$ for a tree $t$}
	\label{alg:treechildcomp}
	\begin{algorithmic}[1]
	\Require $\Gamma_{\geq 1}$ contains no constant, $\Gamma_0$ contains only constants
		\For{$f \in \Gamma_{\geq 1}, 0 < i_1 <i_2 < \cdots < i_\ell \leq \ar(f) =: m, (a_1, a_2, \ldots, a_\ell) \in \Gamma_0^\ell$}
						\State replace each subtree $f(t_1,\ldots,t_m)$ s.t.~$t_{i_j} = a_j$ for $1 \leq j \leq \ell$
						and $t_i \notin \Gamma_0$ for $i \not\in \{i_1, \ldots, i_\ell\}$ 
						\par by $f'(t_1,\ldots,t_{i_1-1},t_{i_1+1},\ldots,t_{i_\ell-1},t_{i_\ell+1},\ldots,t_m)$
						\Comment{$f' \notin \Gamma_{\geq 1} \cup \Gamma_0$}
		\EndFor
	 \end{algorithmic}
\end{algorithm}

In case of the pair compression the situation is a bit more difficult:
observe that in a chain subpattern $abc$ we can compress $ab$ or $bc$ but we cannot do both
in parallel (and the outcome depends on the order of the operations).
However, as in the case of word equations~\cite{wordequations},
parallel $a,b$ pair compressions are possible when we take $a$ and $b$ from disjoint subalphabets
$\Gamma_{1}$ and $\Gamma_2$, respectively.
Those subalphabets are usually a partition of letters present in some tree and so we call them a partition,
even if we do not explicitly say of what.
In this case for each unary letter we can tell whether it should be the parent node
or the child node in the compression step and the result does not depend on the order of the considered pairs,
as long as new letters are outside $\Gamma_1 \cup \Gamma_2$.
This is formalised in the below algorithm \algtreepaircomp$(\Gamma_1,\Gamma_2,t)$,
when $t$ is clear form the context (or unimportant) we refer to it simply as $\Gamma_1,\Gamma_2$ compression
(we list the $\Gamma_1$ and $\Gamma_2$ to stress the dependency of the procedure on them).

\begin{algorithm}[H]
	\caption{\algtreepaircomp$(\Gamma_1,\Gamma_2,t)$: $\Gamma_1, \Gamma_2)$-compression for a tree $t$}	\label{alg:treepaircomp}
	\begin{algorithmic}[1]
	\Require $\Gamma_1, \Gamma_2$ contains only unary letters and are disjoint
		\For{$a \in \Gamma_1$ and $b \in \Gamma_2$}
			\State replace each occurrence of a chain subpattern $ab$ with a fresh letter $c$ \Comment{$c \notin \Gamma_1 \cup \Gamma_2$}
		\EndFor
	 \end{algorithmic}
\end{algorithm}

The main property of the the listed procedures is that they shrink the tree by a constant factor;
to be more precise: the chain compression followed by a $\Gamma_1,\Gamma_2$ child compression
(for a proper choice of partition $(\Gamma_1,\Gamma_2)$) followed by a leaf compression 
applied to a tree $t$ results in a tree $t'$ which is smaller by a constant factor than $t$.
This should be intuitively clear: the leaf compression removes all leaves;
this would halve the size of the the tree if there were no nodes with unary labels.
But such nodes are compressed on their own: first by chain compression and then by the $\Gamma_1$, $\Gamma_2$ compression
and it is known from the earlier work on word equations that for appropriate partition of unary letters
a word is shorten by a constant factor by those two operations~\cite{wordequations}.
The details are presented in the algorithm below and the following theorem.

\begin{algorithm}[H]
	\caption{\algtreecomp$(t)$: Compression of a tree $t$}
	\label{alg:treecomp}	
	\begin{algorithmic}[1]
			\State $\Gamma \gets$ unary letters in $t$
			\State $t \gets \algtreechaincomp(\Gamma,t)$
			\State $\Gamma \gets$ unary letters in $t$
			\State guess partition of $\Gamma$ into $\Gamma_1$ and $\Gamma_2$ \label{partition letters}
			\State $t \gets \algtreepaircomp(\Gamma_1,\Gamma_2,t)$
			\State $\Gamma_0 \gets $ constants in $t$, $\Gamma_{\geq 1} \gets$ other letters in $t$
			\State $t \gets \algtreechildcomp(\Gamma_{\geq 1},\Gamma_0,t)$			
	 \end{algorithmic}
\end{algorithm}

\begin{theorem}
\label{thm: tree size drop}
Consider tree $t$ on which we run $\algtreecomp$ and the obtained tree $t'$.
For some partition $\Gamma_1,\Gamma_2$ it holds that $|t'| < \frac{3|t|}{4}$.
\end{theorem}

As a first technical step we show that when no node is labelled with the same unary letter as its father,
the claim of the theorem holds (note that then the chain compression does not change the tree).
Then it is enough to see that the chain compression cannot increase the size of the tree and that after it
there are no two such nodes, by Lemma~\ref{lem: father and son}.

\begin{lemma}
	\label{lem: tree size drop technical}
Consider a tree $t$ in which no node is labelled with the same unary letter as its father
and a run of \algtreecomp$(t)$. 
Let $t'$ the tree obtained after the $\Gamma_1,\Gamma_2$ compression and $t''$ the tree obtained after the leaf compression.
For some partition $\Gamma_1,\Gamma_2$ it holds that $|t''| < \frac{3|t|}{4}$.
\end{lemma}
\begin{proof}
Note that by the assumption that no node is labelled with the same unary label as its father,
the \algtreechaincomp{} returns the same tree $t$.

Let $n_0$, $n_1$ and $n_{\geq 2}$ denote, respectively, the number of leaves, nodes with only one child and other nodes in $t$,
$n_0'$, $n_1'$ and $n_{\geq 2}'$ the number of such nodes in $t'$,
$n_0''$, $n_1''$ and $n_{\geq 2}''$ in $t''$.
We show that $n_0'' + n_1'' + n_{\geq 2}'' < \frac{3}{4}(n_0 + n_1 + n_{\geq 2})$, which shows the claim.

Clearly
\begin{equation}
	\label{eq: leaves and others}
n_{\geq 2} \leq n_0 - 1 \enspace.
\end{equation}
This is easy to show: except for a root, each vertex has a father, i.e.\ there are $n_0+n_1+n_{\geq 2}-1$ sons,
on the other hand, we can estimate the number of sons by calculating the number of children, which yields that
there are at least $n_1 + 2n_{\geq 2}$ sons, hence $2n_{\geq 2} + n_1 \leq n_0 + n_1 + n_{\geq 2} - 1$, yielding~\eqref{eq: leaves and others}.

Concerning the $\Gamma_1,\Gamma_2$ compression, we first need some notions:
we say that an occurrence of a chain subpattern $ab$ in $t$ is \emph{covered} by $\Gamma_1,\Gamma_2$
if $a \in \Gamma_1$ and $b \in \Gamma_2$.
We claim that there is a partition of unary letters in $t$ (i.e.\ $\Gamma$) into $\Gamma_1,\Gamma_2$
such that at least $\frac{n_1 - c}{4}$ occurrences of two letter chain subpatterns in $t$ are covered by $\Gamma_1,\Gamma_2$,
where $c$ is the number of maximal chains in $t$.
\begin{clm}
	\label{clm: random partition}
Let $\Gamma$ be the set of unary letters in $t'$. There is a partition $\Gamma_1,\Gamma_2$ of $\Gamma$
such that at least $\frac{n_1' - c}{4}$ occurrences of $2$-chains subpatterns in $t$ are covered by $\Gamma_1,\Gamma_2$,
where $c$ is the number of maximal chains in $t$.
\end{clm}
\begin{proof}
Consider a random partition of $\Gamma$ into $\Gamma_1$ and $\Gamma_2$,
which assigns each letter from $\Gamma$ to $\Gamma_1$ or $\Gamma_2$ with equal probability.
Then for a fixed occurrence of a two letter chain subpattern $ab$ in $t$ (note that by the assumption $a \neq b$)
with probability $1/4$ this occurrence is covered by the partition: the probability that $a \in \Gamma_1$ is $1/2$,
probability that $b \in \Gamma_2$ is $1/2$ as well and as $a \neq b$ those two events are independent.
Since there are $n_1 - c$ occurrences of $2$-chain subpatterns in total, expected number of occurrences covered by a partition is
$\frac{n_1 - c}{4}$, so for some partition at least $\frac{n_1 - c}{4}$ occurrences are covered.
\qedhere
\end{proof}

We should now estimate $c$---the number of maximal chains in $t$, it is at most 
\begin{equation}
	\label{eq: bounding number of chains}
	c \leq n_{\geq 2} + \frac{n_0}{2} + \frac{1}{2} \enspace.
\end{equation}
Indeed, consider an occurrence of a maximal chain subpattern in $t$.
Then the node above has a label of arity at least $2$ (unless the occurrence of a maximal chain subpattern includes the root)
while the node below has a label of arity other than $1$.
Summing this up by all chains we get $2c \leq 2 n_{\geq 2} + n_0 + 1$ (the `$+1$' is for the possibility that the root has a unary label),
which yields~\eqref{eq: bounding number of chains}.

Thus for the choice of $\Gamma_1,\Gamma_2$ from Claim~\ref{clm: random partition} there are at least $\frac{n_1 - c}{4}$ $2$-chains compressed,
so the tree is smaller by at least $\frac{n_1 - c}{4}$ nodes, hence
the value of $n_0' + n_1' + n_{\geq 2}'$ is at most
\begin{align}
	\notag
	n_0' + n_1' + n_{\geq 2}'
		&\leq 
	n_0 + n_1 + n_{\geq 2} - \frac{n_1 - c}{4} &\text{from Claim~\ref{clm: random partition}}\\
	\notag
		&=
	n_0 + \frac{3n_1}{4} + n_{\geq 2} + \frac{c}{4}&\text{simplification}\\
	\notag
		&\leq
	n_0 + \frac{3n_1}{4} + n_{\geq 2} + \underbrace{\frac{n_{\geq 2}}{4} + \frac{n_0}{8} + \frac{1}{8}}_{c/4} &\text{from~\eqref{eq: bounding number of chains}}\\
	\label{eq: n'' estimation}
		&=
	\frac{9n_0}{8} + \frac{3n_1}{4} + \frac{5n_{\geq 2}}{4} + \frac{1}{8} &\text{simplification} \enspace .
\end{align}

Lastly, the leaf compression simply removes all $n_{\geq 2}'$ leaves, which is exactly $n_{\geq 2}$
(no leaves are created or removed during the previous compression steps).
Hence
\begin{align*}
	n_0'' + n_1'' + n_{\geq 2}''
		&\leq
	n_0' + n_1' + n_{\geq 2}' - n_0 &\text{leaf compression}\\
		&\leq
	\frac{9n_0}{8} + \frac{3n_1}{4} + \frac{5n_{\geq 2}}{4} + \frac{1}{8}  - n_0 &\text{from~\eqref{eq: n'' estimation}}\\
		&\leq
	\frac{n_0}{8} + \frac{3n_1}{4} + \frac{5n_{\geq 2}}{4} + \frac{1}{8} &\text{simplification}\\
		&=
	\frac{3}{4} \Big(n_0 + n_1 + n_{\geq 2}\Big) + \Big(-\frac{5n_0}{8} + \frac{n_{\geq 2}}{2} + \frac{1}{8}\Big) &\text{desimplification}\\
		&<
	\frac{3}{4} \Big(n_0 + n_1 + n_{\geq 2}\Big) &\text{from~\eqref{eq: leaves and others}}
	\enspace .
\end{align*}
\qedhere
\end{proof}

Now the proof of Theorem~\ref{thm: tree size drop} follows.

\begin{proof}[proof of Theorem~\ref{thm: tree size drop}]
The tree obtained from $t$ after the chain compression is clearly at most as large as before,
so it is enough to show that the application of $\Gamma_1$, $\Gamma_2$ compression followed by the leaf compression
reduces the size of tree by at least one fourth.
This is shown in Lemma~\ref{lem: tree size drop technical}, ut with the additional assumption that there is no node
labelled with the same unary letter as its father.
But by Lemma~\ref{lem: father and son} we know that this assumption holds after the chain compression.
Which ends the proof.
\qedhere
\end{proof}

The essential part of the paper is showing, how to modify the equation so that the
compression steps can be performed directly on the context equation.

In the next sections the following observation, which bounds the maximal arity of the letters introduced during the compression steps,
proves useful

\begin{lemma}
	\label{lem: keeping arity low}
If the maximal degree of nodes in $t$ is $k$ then in $t' = \algtreecomp(t)$ the maximal degree of a node is also at most $k$.
\end{lemma}
\begin{proof}
Observe that the chain compression replaces chain of unary nodes with a single unary node.
Similarly, $\Gamma_1,\Gamma_2$ compression replaces chains of length two with single unary letters.
Lastly, leaf compression can only reduce the arity of a node (or keep it the same).
\end{proof}

\section{Context unification}
\label{sec:context unification}
In this section we (more formally) define the problem of context unification and the notions necessary to state the problem.
The presentation here is based on~\cite{contexttwovar}.

Recall that $\Sigma$ is the set of letters used as labels for nodes in trees.
By $\Omega$ we denote a special constant outside $\Sigma$ (and no letter added to $\Sigma$ may be equal to $\Omega$).
\contextvar{} denotes an infinite set of context variables $X$, $Y$, $Z$, \ldots.
We also use individual variables $x$, $y$, $z$, \ldots taken from \variables.

\begin{definition}[cf.~{\cite[Definition~2.1]{contexttwovar}}]
A \emph{ground context} is a ground ($\Sigma \cup \{ \Omega\}$)-term $t$, where $\ar(\Omega) = 0$,
that has exactly one occurrence of the constant $\Omega$.
The ground context $\Omega$ is the \emph{empty ground context}.
\end{definition}

The intuition of the symbol $\Omega$ is that it is a `hole' and that one should replace this hole
with a ground term to obtain a proper ground term.

Given a ground context $s$ and a ground term/context $t$ we write $st$ for the ground
term/context that is obtained from $s$ when we replace the occurrence of $\Omega$ in $s$ by $t$.
(This form of composition is associative.)
In the same spirit, when $a$ is a unary letter, we usually write $at$ to denote $a(t)$.

\begin{definition}[cf.~{\cite[Definition~2.2]{contexttwovar}}]
The \emph{terms} over $\Sigma$, $\variables$, $\contextvar$ are ground terms with alphabet $\Sigma \cup \variables \cup \contextvar$
in which we extend $\ar$ to $\variables \cup \contextvar$ by $\ar(X) = 1$ and $\ar(x) = 0$ for each $x \in \variables$ and $X \in \contextvar$.

The \emph{context terms} are ground context terms over $\Sigma \cup \variables \cup \contextvar \cup \{ \Omega\}$
with exactly one occurrence of $\Omega$ and $\ar$ extended to $\variables \cup \contextvar$ as above and $\ar(\Omega) = 0$.

A \emph{context equation} is an equation of the form $u = v$ where both $u$ and $v$ are terms.
\end{definition}

We call the letters from $\Sigma$ that occur in a context equation the \emph{explicit letters}
and talk about explicit occurrences of letters in a context equation.

\subsection{Solutions}
We are interested in the solutions of the context equations, i.e.\ substitutions that replace variables
with ground terms and context variables with ground contexts, such that a formal equality $u = v$
is turned into a true equality of ground terms.
More formally:

\begin{definition}[cf.~{\cite[Definition~2.3]{contexttwovar}}]
\label{def: solution}
A \emph{substitution} is a mapping \solution{} that assigns a ground context \sol X to
each context variable $X \in \contextvar$ and a ground term \sol x to each variable $x \in \variables$.
The mapping \solution{} is naturally extended to arbitrary terms as follows:
\begin{itemize}
	\item $\sol a := a$ for each constant $a \in \Sigma$;
	\item $\sol {f(t_1, \ldots, t_n)} := f(\sol{t_1}, \ldots, \sol{t_m})$ for an $m$-ary $f \in \Sigma$;
	\item $\sol {Xt} := \sol X \sol t$ for $X \in \variables$.
\end{itemize}
A substitution \solution{} is a \emph{solution} of the context equation $u = v$ if $\sol u = \sol v$.
A solution \solution{} is \emph{size-minimal}, if for every other solution $\solution'$ it holds that $|\sol u| \leq |\solution'(u)|$.
A solution \solution{} is non-empty if $\sol X \neq \Omega$ for each $X \in \variables$ from the context equation $u = v$.
\end{definition}

In the following, we are interested only in non-empty solutions.
Notice that this is not restricting, as for the input instance we can guess, which context variables have empty substituion in the solution
and remove them.

For a ground term \sol u and an occurrence of a letter $a$ in it we say that this occurrence \emph{comes from} $u$ it is was obtained as \sol a
in Definition~\ref{def: solution} and that it comes from $X$ (or $x$) if it was obtained from \sol X (or \sol x, respectively)
in Definition~\ref{def: solution}.

Let the maximal arity of letters in $\Sigma$ be $k$.
We claim that without loss of generality we may assume that for each $k' \leq k$
the $\Sigma$ contains a letter of arity $k'$; this is formalised in the following lemma.

\begin{lemma}
\label{lem: all arities}
Let $\Sigma$ be a signature (that contains a constant) and $f \in \Sigma$ be the letter of maximal arity ($k$) in $\Sigma$ and $u = v$
be a context equation.
If $u = v$ has a solution $\solution'$ over a signature $\Sigma'$ such that each letter in $\Sigma' \setminus \Sigma$
is not in $u = v$ and has arity at most $k$ then $u = v$ also has a solution, which is at most $k$ times larger.
\end{lemma}
Note that if a signature does not contain a constant then no terms can be build using it.
\begin{proof}
Let $f$ be a letter of arity $k$ in $\Sigma$ and $a$ any constant in $\Sigma$.
Let $h$ be any letter used by $\solution'$ which is neither in $u = v$ nor in $\Sigma$, let $m = \ar(h)$.
We change $\solution'$ by replacing each $h(t_1,t_2,\ldots,t_m)$ by $f(t_1,t_2,\ldots,t_m,\underbrace{a, a, \ldots, a}_{k - m \text{ times}})$.
It is easy to see that the new substitution is also a solution: as $h$ is not used in the equation,
each of its occurrences in \sol u or \sol v comes from \sol X or \sol x and we replace each such occurrence with
$f(\ldots \underbrace{a, a, \ldots, a}_{k - m \text{ times}})$.
Iterating over all $h$ yields the claim.
\qedhere
\end{proof}
In the following we use the Lemma~\ref{lem: all arities} implicitly --- we always assume that $\Sigma$ contains
letters of arity $0$, $1$, \ldots, $k$.
Note that this assumption may influence the size of the size-minimal solution (as it may decrease sizes of some solutions),
we disregard this small effect and assume that the signature preprocessing is done before any algorithm is run.

\subsection{Properties of solutions}
In case of word equations, size-minimal solutions are considered mainly because one can bound the \emph{exponent of periodicity} for them,
below we recall a known similar fact for context equations.

\begin{lemma}[Exponent of periodicity bound~\cite{contextexponent}]
	\label{lem: exponent of periodicity}
Let \solution{} be a size-minimal solution of a context equation $u = v$. Suppose that \sol X (or \sol x) can be written as $t s^m t'$,
where $t, s, t'$ are ground context terms (or $t'$ is a ground term, respectively).
Then $m = 2^{\Ocomp(|u|+|v|)}$.
\end{lemma}

We use Lemma~\ref{lem: exponent of periodicity} only for the case when $s$ is a unary letter, for which the proof simplifies significantly
and is essentially the same as in the case of word equations~\cite{wordequations} (which is a simplification of the general
bound on the exponent of periodicity by Ko\'scielski and Pacholski~\cite{Koscielski}).

Furthermore in case of word equation the minimality of solution is used also for another purpose:
whenever a letter $a$ occurs in the minimal solution, it needs to occur also in the equation~\cite{PlandowskiICALP}.
\begin{lemma}[\cite{PlandowskiICALP}]
\label{lem: over a cut}
Let $\solution$ be a length minimal solution of a word equation $u = v$.
If a letter $a$ occurs in \sol u then it occurs also in $u$ or $v$.
\end{lemma}

The property from Lemma~\ref{lem: over a cut}
is very useful, as we can still deduce the set of letters used
by minimal solutions simply by looking at the equation,
in particular we can restrict ourselves to $\Sigma$ such that $|\Sigma| \leq |u| + |v|$.

However, this is \emph{not} the case of the context-equations: an equation $X(a)=Y(b)$ over a signature $\{f,a,b \}$,
with $f$ being binary and $a$, $b$ being constants, has a solution (which is easily seen to be size-minimal) $\sol X = f(\Omega,b)$
and $\sol Y = f(a,\Omega)$ and in fact each solution needs to use $f$, which does not occur in the context equation.
Still, a very similar property holds for context equations:
We say that for solutions \solution{} and $\solution'$ of \rownanie{} the $\solution'$  is a \emph{simpler equivalent} of \solution{}
if $\solution'$ is obtained by applying a letter homomorphism on top of \solution{}
(i.e.\ we exchange each letter $a$ in \sol X and \sol x by some fixed $h(a)$ of the same arity).
Then following simple lemma gives the a relatively close approximation of Lemma~\ref{lem: over a cut} for the case of context equations:
\begin{lemma}
	\label{lem: almost over a cut}
Consider a context equation $u = v$ over a signature $\Sigma$, such that the maximal arity of letters in $\Sigma$ is $k$.
Then for every solution \solution{} there is a simpler equivalent $\solution'$
such that for each $k' = 0, 1, 2, \ldots, k$ it uses at most one letter of arity $k'$ that is not used in \rownanie.

Furthermore, for each solution \solution{} there exists a solution $\solution'$ that uses only unary letters that are present in $u = v$
and $|\solution'(u)| \leq |\sol u|$.
\end{lemma}
Note that the usage of Lemma~\ref{lem: almost over a cut} is similar to the one of Lemma~\ref{lem: over a cut}:
whenever we have an equation $u = v$ over some large signature $\Sigma$ (with maximal arity $k$)
we can remove from $\Sigma$ all (except for $k+1$) letters that do not occur in the equation $u = v$
while preserving satisfiability.

Note also that Lemma~\ref{lem: almost over a cut} that among the size-minimal solutions there exist one
that does not any unary letters not present in the equation $u = v$.
\begin{proof}
Let \solution{} be a solution of \rownanie.
Fix some $0 \leq k' \leq k$ and all letters $f,g,h,\ldots$ that occur in \sol u but not in $u = v$.
Replace each occurrence of $g,h\ldots$ in \sol X and \sol x (for each context variable $X$ and each variable $x$) with $f$.
We claim that the new substitution $\solution'$ is also a solution:
since the letters $g,h,\ldots$ did not occur in the equation $u = v$ then the only way that they could occur in \sol u and \sol v
is from \sol X or \sol x.
But in all of those variables and context variables we uniformly replaced $g,h,\ldots$ with $f$.
So the equality $\solution'(u) = \solution'(v)$ still holds.
Iterating this argument for $k' = 0, 1, 2, \ldots, k$ yields the first claim.

To show the second claim consider a solution \solution{} and a unary letter $a$ it uses which is not in \rownanie.
Consider a new solution $\solution'$ which is obtained from \solution{} by deleting each $a$, i.e.\ replacing each subterm $a^\ell t$
where $a^\ell$ is $a$-maximal with $t$ (since $a$ is a unary letter, this is a valid operation on terms).
Since all $a$ in \sol u and \sol v came from the \solution,
the $\solution'(u)$ is obtained from \sol u by deleting all $a$s, similarly
$\solution'(v)$ is obtained from \sol v by deleting all $a$s.
Hence $\solution'(u) = \solution'(v)$, which shows that indeed $\solution'$ is a solution of \rownanie.
Iterating over all unary letters $a$ in \solution{} that are not in \rownanie{} yields the second claim.
\qedhere
\end{proof}

The usage of Lemma~\ref{lem: almost over a cut} is as follows:
when we want to perform the compression steps we need to know what are the letters used by (some) solution.
Lemma~\ref{lem: almost over a cut} states that without loss of generality we can restrict ourselves
to letters present in the equation and one letter for each arity.
Furthermore, when we are concerned only with unary letters,
we can consider only letters that are present in the equation.

\section{Compression of non-crossing subpatterns}
\label{sec: noncrossing compression}
In this section we adapt the compressions from Section~\ref{sec:trees} to the case when the terms are given
implicitly, i.e.\ as a solution of a context equation.
To this end we identify cases, in which performing such a compression is easy and those in which it is hard
and show how to make the compression in the easy cases.
In the next section we present how to transform the difficult cases to the easy ones.
Finally, in Section~\ref{sec: main algorithm} we wrap everything up and present the algorithm for context-unification together
with the space-usage analysis.

However, before stating the procedures that transform the equations,
we formalise the notions about their correctness.
This might be a little non-obvious, as our procedures in general make non-deterministic choices.

\subsection{Soundness and Completeness}
The intuition of the correctness of a non-deterministic procedure is clear:
if the context equation is satisfiable then for some non-deterministic choices
we should transform it to a (simpler) satisfiable instance.
If it is unsatisfiable, we can transform it only to a non-satisfiable instance, regardless of the non-deterministic choices.

\begin{definition}
A (nondeterministic) procedure is \emph{sound},
when given a unsatisfiable word equation $u = v$
it cannot transform it to a satisfiable one, regardless of the nondeterministic choices;
such a procedure is \emph{complete},
if given a satisfiable equation $u = v$ 
for some nondeterministic choices it returns a satisfiable equation $u' = v'$.
\end{definition}
Observe, that a composition of sound (complete) procedures is also sound (complete, respectively)

A very general class of operations is sound:
\begin{lemma}
\label{lem:preserving unsatisfiability}
The following operations are sound:
\begin{enumerate}
	\item \label{unsat 1} Replacing all occurrences of a context variable $X$ (variable $x$) with $tX$ ($tx$, respectively) throughout the $u = v$,
	where $t$ is a context term.
	\item \label{unsat 2} Replacing all occurrences of a context variable $X$  with $Xt$ throughout the $u = v$ where $t$ is a context term. \label{second case}
	\item \label{unsat 25} Replacing all occurrences of a variable $x$ with a ground term $t$.
	\item \label{unsat 3} $(f,i,a)$ leaf compression performed on $u = v$.
	\item \label{unsat 4} $a,b$ pair compression performed on $u = v$.
	\item \label{unsat 5} $a$-maximal chain compression performed on $u = v$.
\end{enumerate}
\end{lemma}
Note that a context term may include variables, letters of large arity etc.
However, we use Lemma~\ref{lem:preserving unsatisfiability} in a very restricted scenario,
in which we use only constants unary letters as context terms
(except case~\ref{second case}, in which we replace $X$ with $X(f(x_1,x_2,\ldots,x_{i-1},\Omega,x_{i+1},\ldots,x_m))$).

\begin{proof}
The proof follows a simple principle: if the obtained equation $u' = v'$ has a solution $\solution'$
then we can define a solution \solution{} of the original context equation by reversing the performed operation.

In~\ref{unsat 1}, if $\solution'$ is a solution of the new equation then $\sol X = \solution'(t) \solution'(X)$ is a solution
(the same holds for $x$).

Similarly, in~\ref{unsat 2}, if $\solution'$ is a solution of the new equation then $\sol X = \solution'(X)\solution'(t)$
is a solution of the original equation.

In~\ref{unsat 25} if $\solution'$ is a solution of the new equation, we define \solution{} i the same way,
but set $\sol x = t$.

In~\ref{unsat 3}, let $f'$ denote the letter that replaced $f$ with child $a$ at positions $i$
during the $(f,i,a)$ leaf compression.
Let $\solution'$ be a solution of the new equation,
we define a solution \solution: if $\solution'(X)$ contains the occurrences of a letter $f'$,
then we replace the whole subterm $f'(t_1,t_2,\ldots,t_{i-1},t_{i+1},\ldots,t_k)$ in $\solution'(X)$
with $f(t_1,\ldots,t_{i-1},a_1,t_{i+1},\ldots,t_k)$, the same is done for \sol x.

In case~\ref{unsat 4}, if the letter $c$ occurs in \sol X or \sol x then we replace it with a chain pattern $ab$.

Similarly, in the last case \solution{} is obtained from $\solution'$ by replacing each occurrence of a letter $a_\ell$ with a chain $a^\ell$
(for all $\ell$).

It is easy to see that all those operations define a valid solution of the original equation.
\qedhere
\end{proof}

\subsection{Non-crossing partitions and their compression}
We begin the considerations with the $\Gamma_1,\Gamma_2$ compression as it is the easiest to explain and the intuition behind is most apparent.

Consider a context equation $u = v$ and a solution \solution.
Suppose that we want to perform the $\Gamma_1,\Gamma_2$ compression on \sol u and \sol v,
i.e.\ we want to replace each occurrence of a chain subpattern $ab \in \Gamma_1\Gamma_2$ with a fresh unary letter $c$.
Such replacement is easy, when the occurrence of $ab$ subpattern comes from the letters in the equation
or from \sol X (or \sol x) for some context variable $X$ (or a variable $x$, respectively):
in the former case we modify the equation be replacing the subpattern $ab$ with $c$, in the latter the modification is done implicitly
(i.e.\ we replace the subpattern $ab$ in \sol X or \sol x with $c$).
The problematic part is with the $ab$ chain subpattern that is of neither of those forms,
as they `cross' between \sol X (or \sol x) and some letter outside this \sol X (or \sol x).
This is formalised in the below definition.

\begin{definition}
For an equation $u = v$ and a non-empty substitution \solution{} we say that an occurrence of a chain subpattern $ab$ in \sol u (or \sol v) is
\begin{description}
	\item[explicit with respect to \solution] the occurrences of both $a$ and $b$ come from explicit letters $a$ and $b$ in $u = v$;
	\item[implicit with respect to \solution] the occurrences of both $a$ and $b$ come from \sol x (or \sol X);
	\item[crossing with respect to \solution] otherwise.
\end{description}

We say that $ab$ is a \emph{crossing pair} with respect to \solution{} if they have at least one crossing occurrence
with respect to \solution.
Otherwise $ab$ is a \emph{non-crossing pair} (with respect to \solution).
\end{definition}

Unless explicitly written, we consider only crossing/noncrossing pairs $ab$ in which $a \neq b$.

The notions of a crossing chain subpattern can be defined in a more operational manner:
for a non-empty substitution \solution{} by \emph{first letter} of \sol X (\sol x) we denote the topmost-letter in \sol X (\sol x, respectively),
by the \emph{last letter} of \sol X we denote the function symbol that is the father of $\Omega$ in \sol X.
Then it is easy to see that $ab$ is crossing with respect to \solution{} if and only if one of the following conditions hold
for some context variables $X, Y$ (or variable $y$):
\begin{itemize}
	\item $aX$ (or $ax$) is a chain subpattern in $u = v$ and $b$ is the first letter of \sol X (or \sol x, respectively) \emph{or}
	\item $Xb$ is a chain subpattern in $u = v$ and $a$ is the last letter of \sol X \emph{or}
	\item $XY$ (or $Xy$) is a chain subpattern in $u = v$, $a$ is the last letter of \sol X and $b$ the first letter of \sol Y
	(\sol y, respectively).
\end{itemize}
These conditions prove to be useful afterwards.

Since we perform several pair compression in one go, we generalise a definition of a crossing pairs to partitions:
\begin{definition}
A partition $\Gamma_1$, $\Gamma_2$ of $\Gamma$ is \emph{non-crossing} (with respect to a solution \solution)
if there is no pair $ab$ with  $a \in \Gamma_1$ and $b \in \Gamma_2$ such that $ab$ is a crossing pair (with respect to \solution);
otherwise it is \emph{non-crossing} with respect to \solution.
\end{definition}

When a partition $\Gamma_1,\Gamma_2$ is non-crossing with respect to a solution \solution,
we can simulate the $\algtreepaircomp(\Gamma_1,\Gamma_2,\sol u)$ on $u = v$
simply be performing the $\Gamma_1$, $\Gamma_2$ compression on the explicit letters in the equation:
then occurrences of $ab$ that come from explicit letters are compressed, the ones that come from \sol X and \sol x
are compressed by changing the solution and there are no other possibilities.
To be more precise we treat the equation $u = v$ as a term over $\Sigma \cup \variables \cup \contextvar \cup \{=\}$
(imagine $u$ and $v$ as children of the root labelled with `$=$', which has arity $2$)
and apply the $\Gamma_1$, $\Gamma_2$ pair compression on this tree,
we refer tot this operation as $\algtreepaircomp(\Gamma_1,\Gamma_2,\rownaniep)$
(note that context variables are not in $\Gamma_1$, nor in $\Gamma_2$
while variables as well as `$=$' have arity other than $1$ so they cannot be compressed either).

\begin{algorithm}[H]
  \caption{$\algpaircompncr(\Gamma_1,\Gamma_2,\rownaniep)$ $\Gamma_1$, $\Gamma_2$ compression for a non-crossing  partition $\Gamma_1, \Gamma_2$ \label{alg:paircomp}}
  \begin{algorithmic}[1]
	\Require $\Gamma_1, \Gamma_2$ contain only unary letters and are a non-crossing partition
		\State run $\algtreepaircomp(\Gamma_1,\Gamma_2,`u=v')$ \Comment{Treat, $u = v$ as a tree}
		\par 
		\Comment{Variables and context variables are not compressed}
 \end{algorithmic}
\end{algorithm}

\begin{lemma}
\label{lem: algpaircomp noncrossing}

$\algpaircompncr(\Gamma_1,\Gamma_2,\rownaniep)$ is sound.

If $u = v$ has a solution \solution{} such that $\Gamma_1$, $\Gamma_2$ is a non-crossing partition with respect to \solution{} then
$\algpaircompncr(\Gamma_1,\Gamma_2,\rownaniep)$ is complete, to be more precise,
the returned equation $u' = v'$ has a solution $\solution'$ such that $\solution'(u') = \algtreepaircomp(\Gamma_1,\Gamma_2,\sol u)$.
\end{lemma}
\begin{proof}
Note that while $\algpaircompncr(\Gamma_1,\Gamma_2,\rownaniep)$ applies several $ab$ pair compression for $ab \in \Gamma_1\Gamma_2$
in parallel, since $\Gamma_1$ and $\Gamma_2$ are disjoint, we can in fact think that they are done sequentially.
Furthermore, when done `sequentially', the pairs in $\Gamma_1\Gamma_2$ do not become non-crossing, as we do not introduce any new letters from
$\Gamma_1$, nor $\Gamma_2$ to the equation, nor to the solution.
For each such compression we use Lemma~\ref{lem:preserving unsatisfiability} to show that it is sound
and so also $\algpaircompncr(\Gamma_1,\Gamma_2,\rownaniep)$ is.

No, concerning the completeness.
Suppose that $u = v$ has a solution \solution{} such that $\Gamma_1,\Gamma_2$ is a non-crossing partition with respect to \solution.
We define a substitution $\solution'$ for the obtained equation $u' = v'$
such that $\solution'(u') = \algtreepaircomp(\Gamma_1,\Gamma_2,\sol u)$ and symmetrically
$\solution'(v') = \algtreepaircomp(\Gamma_1,\Gamma_2,\sol v)$.
Since $\sol u = \sol v$ this shows that $\solution'$ is indeed a solution of $u' = v'$ and so the second claim of the lemma holds.

The definition is straightforward: $\solution'(X)$ is obtained by performing the $\Gamma_1,\Gamma_2$ compression
on \sol X (the $\solution'(x)$ is defined in the same way)
formally $\solution'(X) = \algtreepaircomp(\Gamma_1,\Gamma_2,\sol X)$ (note that $\Omega$ is not in $\Gamma_1$, nor in $\Gamma_2$ 
and so it is not replaced).

Consider $a \in \Gamma_1$ with child labelled with $b \in \Gamma_2$ in \sol u.
Consider where this chain subpattern $ab$ comes from:
\begin{description}
	\item[they both come from explicit letters]
	Then $\algpaircompncr(\Gamma_1,\Gamma_2,\rownaniep)$ will perform the $\Gamma_1,\Gamma_2$ compression on them,
	i.e.\ replace them with a letter $c$.
	\item[they both come from \sol X or \sol x] Then this occurrence of $ab$ is replaced by the definition of $\solution'$.
	\item[one of them comes from an explicit letter and one from \sol X or \sol x]
	But then $\Gamma_1,\Gamma_2$ is a crossing partition with respect to \solution, contradicting the assumption.
\end{description}
As the argument applies to every occurrence of chain subpattern $ab \in \Gamma_1\Gamma_2$,
this shows that $\solution'(u') = \algtreepaircomp(\Gamma_1,\Gamma_2,\sol u)$, which ends the proof of the lemma.
\qedhere
\end{proof}

\subsection{Non-crossing $a$-maximal chains and their compression}
Similarly, consider a context equation $u = v$ and a solution \solution.
Suppose that we want to perform the $a$-maximal chain compression on \sol u and \sol v.
Then all occurrences of $a$-maximal chains are to be replaced with new unary letters.
Such replacement is easy, when the chain is a chain subpattern of the equation or is a chain subpattern
of \sol X (or \sol x) for some context variable $X$ (or a variable $x$, respectively).
The problematic part is with the occurrences that are of neither of those forms,
as they `cross' between \sol X (or \sol x) and another subtree.
This is formalised in the below definition.

\begin{definition}
For an equation $u = v$ and a substitution \solution{} we say that an occurrence of an $a$-maximal chain subpattern $a^\ell$ in \sol u
(or \sol v) is
\begin{description}
	\item[explicit with respect to \solution] if this occurrence comes wholly from $u$ (or $v$), i.e.\ it is a chain subpattern of $u$ (or $v$);
	\item[implicit with respect to \solution] if this occurrence comes wholly from \sol X or \sol x, i.e.\ it is a chain subpattern of \sol X or \sol x;
	\item[crossing with respect to \solution] otherwise.
\end{description}
We say that $a$ has a \emph{crossing chain} if there is at least one occurrence of a crossing $a$-maximal chain subpattern.
Otherwise, $a$ \emph{has no crossing chain}.
\end{definition}

As in the case of $\Gamma_1$, $\Gamma_2$, it is easy to see that $a$ has a crossing chain with respect to a non-empty solution \solution{}
if and only if one of the following holds
for some context variables $X, Y$ (or variable $y$):
\begin{itemize}
	\item $ax$ (or $aX$) is a chain subpattern in $u = v$ and the first letter of \sol x (\sol X, respectively) is $a$;
	\item $Xa$ is a chain subpattern in $u = v$ and $a$ is the last letter of $\sol X$;
	\item $XY$ (or $Xy$) is a chain subpattern in $u = v$ and $a$ is the last letter of $\sol X$ as well as the first letter of \sol Y (or \sol y).
\end{itemize}

When no unary letter (from $\Gamma)$ has a crossing chain to simulate the chain compression on the context equation
we perform the \algtreechaincomp{} on the explicit letters, treating the context equation as a tree,
similarly as in the case of the $\Gamma_1$, $\Gamma_2$ compression.

\begin{algorithm}[H]
	\caption{$\algchaincompncr(\Gamma,\rownaniep)$: Compressing chains when there is no crossing chain}
	\label{alg:chaincomp}
	\begin{algorithmic}[1]
	\Require $\Gamma$ contains only unary letters, there are no crossing chains for letters in $\Gamma$
		\State run $\algtreechaincomp(\Gamma,\rownaniep)$ \Comment{Treat, $u = v$ as a tree}
		\par \Comment{Context variables are not compressed}
	\end{algorithmic}
\end{algorithm}

\begin{lemma}
\label{lem: algchaincomp noncrossing}
$\algchaincompncr$ is sound.

If $u = v$ has a solution \solution{} such that no letter in $\Gamma$ has a crossing chain
then it is complete, to be more precise,
the returned equation $u' = v'$ has a solution $\solution'$ such that $\solution'(u') = \algtreechaincomp(\Gamma,\sol u)$.
\end{lemma}
The proof is essentially the same as in Lemma~\ref{lem: uncrossing partition} and so it is omitted.

\subsection{Non-crossing father-leaf pairs and their compression}
Suppose now that given a context equation $u = v$ with a solution \solution{} we would like to perform
leaf compression on \sol u and \sol v.
To this end we need to identify each $f \in \Gamma_{\geq 1}$ and its children in $\Gamma_0$
and replace them accordingly. Again, this is easy if each occurrence of such a subpattern comes either from explicit letters
in $u = v$ or wholly from \sol X (or \sol x).
In such a case we proceed similarly as in the case of $\Gamma_1,\Gamma_2$-compression and chain compression
and treat the $u = v$ as a tree and perform the leaf compression on it.
We are left to identify the cases in which this indeed properly simulates the leaf compression,
which are similar to those in $\Gamma_1$, $\Gamma_2$ compression.

\begin{definition}
Let $\ar(f) \geq 1$ and $\ar(a) = 0$ and consider a subpattern consisting of $f$ with a child $a$
(on some position $i \leq \ar(f)$).
For an equation $u = v$ and a substitution \solution{} we say
we say that an occurrence of such a subpattern is
\begin{description}
	\item[explicit with respect to \solution] if both the occurrence of $f$ and $a$ come from explicit letters in $u$ (or $v$);
	\item[implicit with respect to \solution] if both the occurrence of $f$ and $a$ come from some \sol X or \sol x;
	\item[crossing with respect to \solution] otherwise.
\end{description} 

Then $(f,a)$ is a \emph{crossing parent-leaf pair} in $u = v$ with respect to \solution{}
if it has at least one crossing occurrence in $u = v$ with respect to \solution.
Otherwise it is \emph{noncrossing} with respect to \solution.
\end{definition}

It is easy to observe that
there is such a crossing father-leaf pair $f,a$ (with respect to a non-empty \solution) if and only if
one of the following holds for some context variables $X$ and $y$
\begin{itemize}
	\item $f$ with a son $y$ is a subpattern in $u = v$ and $\sol y = a$ \emph{or}
	\item $Xa$ is a subpattern in $u = v$ and the last letter of \sol X is $f$ \emph{or}
	\item $Xy$ is a subpattern in $u = v$, $\sol y = a$ and $f$ is the last letter of \sol X.
\end{itemize}

When there is no crossing father-leaf pair $(f,a)$ for $f \in \Gamma_{\geq 1}$ and $a \in \Gamma_0$ then to simulate leaf compression
on \sol u and \sol v it is enough to perform it on the equation, treating it as a tree.

\begin{algorithm}[H]
	\caption{$\algchildcompncr(\Gamma_{\geq 1},\Gamma_0,\rownaniep)$: Leaf compression when there is no crossing father-leaf pair}
	\label{alg:leafcomp}
	\begin{algorithmic}[1]
	\Require $\Gamma_{\geq 1}$ contains no constant, $\Gamma_0$ contains only constants,\par
	there is no crossing father-leaf pair $(f,a)$ with $f \in \Gamma_{\geq 1}$ and $a \in \Gamma_0$
		\State run $\algtreechildcomp(\Gamma_{\geq 1},\Gamma_0,\rownaniep)$ \Comment{Treat, $u = v$ as a tree}
		\par \Comment{Context variables and variables are not compressed}
	\end{algorithmic}
\end{algorithm}

\begin{lemma}
\label{lem: algchildcomp noncrossing}
$\algchildcompncr$ is sound.

If $u = v$ has a solution \solution{} such that there is no crossing father-leaf pair $(f,a)$ with $f \in \Gamma_{\geq 1}$ and $a \in \Gamma_0$
in $u = v$ with respect to \solution{} then it is complete, more precisely,
the returned equation $u' = v'$ has a solution $\solution'$ such that $\solution'(u') = \algtreechildcomp(\Gamma_{\geq 1},\Gamma_0,\sol u)$.
\end{lemma}
The proof is essentially the same as in Lemma~\ref{lem: uncrossing partition} and so it is omitted.

\section{Uncrossing}
\label{sec:uncrossing}
In general, one cannot assume that an arbitrary partition $\Gamma_1$, $\Gamma_2$ is noncrossing,
similarly we cannot assume that there are no crossing chains nor crossing father-leaf pairs.
However, for a fixed partition $\Gamma_1$, $\Gamma_2$ and a solution \solution{}
we can modify the instance so that this fixed partition becomes non-crossing with respect to a solution $\solution'$
(that corresponds to \solution{} of the original equation);
similarly, given an equation $u = v$ we can turn it into an equation that has no letters with a crossing chain with respect to a solution
$\solution'$ of the new equation;
lastly, for $\Gamma_{\geq 1}$ and $\Gamma_0$ we can modify the instance so that no father-leaf pair $(f,a)$ with $f \in \Gamma_{\geq 1}$
and $\Gamma_0$ is crossing with respect to $\solution'$.
Those modifications are the cornerstone of our main algorithm, as they allow compression to be performed directly on the equation,
regardless of how the solution actually look like.

\subsection{Uncrossing partitions}
We begin with showing how to turn a partition into a non-crossing one.
Recall that $\Gamma_1,\Gamma_2$ is a crossing partition (with respect to a non-empty \solution) if and only if
for some $ab \in \Gamma_1\Gamma_2$ one of the following holds for some context variables $X, Y$ (or variable $y$),
we assume here that \solution{} is non-empty
\begin{enumerate}[({CP}1)]
	\item \label{cr 1} $aX$ (or $ax$) is a chain subpattern in $u = v$ and $b$ is the first letter of \sol X (or \sol x, respectively) \emph{or}
	\item \label{cr 2} $Xb$ is a chain subpattern in $u = v$ and $a$ is the last letter of \sol X \emph{or}
	\item \label{cr 3} $XY$ (or $Xy$) is a chain subpattern in $u = v$, $a$ is the last letter of \sol X and $b$ the first letter of \sol Y
	(\sol y, respectively).
\end{enumerate}
In each of those cases it is easy to modify the instance so that $ab$ is no longer a crossing pair:
\begin{itemize}
	\item In \CPref{1} we \emph{pop up} the letter $b$: we replace $X$ ($x$) with $bX$ ($bx$, respectively).
	In this way we also modify the solution \sol X (\sol x) from $\sol X = bt$ ($\sol x = bt$, respectively)
	to $\solution'(X) = t$ ($\solution'(x) = t$, respectively).
	If $\solution'(X)$ is empty, we remove $X$ from the equation.
	\item In \CPref{2} we \emph{pop down} the letter $a$:
	we replace each occurrence of $X$ with $Xa$.
	In this way we implicitly modify $\sol X = sa\Omega$ to $\solution'(X) = s$.
	If $\solution'(X)$ is empty, we remove $X$ from the equation.
	\item The case \CPref{3} is a combination of the two cases above, in which we need to pop-down from $X$ and pop-up from $Y$ (or $y$).
\end{itemize}
It is easy to observe that this procedure can be performed on all $ab \in \Gamma_1\Gamma_2$ in parallel,
as presented in the algorithm below.

\begin{algorithm}[H]
  \caption{$\algpop(\Gamma_1,\Gamma_2,\rownaniep)$ \label{alg:leftpop}}
  \begin{algorithmic}[1]
	\For{$X \in \contextvar$}
		\State let $a$ be the last letter of \sol X \label{guess last letter}\Comment{Guess}
		\If{$a \in \Gamma_1$}
			\State replace each occurrence of $X$ in $u = v$ by  $Xa$ \label{rightpop}\\
			\Comment{Implicitly change $\sol X = sa\Omega$ to $\sol X = s$}
			\If{$\sol X$ is empty} \Comment{Guess}
				\State remove $X$ from $u = v$: replace each $X(s)$ in by  $s$
			\EndIf
		\EndIf
	\EndFor	
	\For{$X \in \contextvar$ or $x \in \variables$}
		\State let $b$ be the first letter of \sol X (or \sol x) \label{guess first letter}\Comment{Guess}
		\If{$b \in \Gamma_2$} 
			\State replace each occurrence of $X$ in $u = v$ by $bX$ (or $x$ with $bx$)\label{leftpop}\\
			\Comment{Implicitly change $\sol X = bs$ to $\sol X = s$ or $\sol x = bt$ to $\sol x = t$}
			\If{$\sol X$ is empty} \Comment{Guess}
				\State remove $X$ from $u = v$: replace each $X(s)$ in by  $s$
			\EndIf
		\EndIf
	\EndFor
  \end{algorithmic}
\end{algorithm}

We show that if $u = v$ has a solution \solution{} then for appropriate non-deterministic choices $\algpop(\Gamma_1,\Gamma_2,\rownaniep)$
returns an equation $u' = v'$ that has a solution $\solution'$ such that $\Gamma_1,\Gamma_2$ is non-crossing with respect to $\solution'$,
furthermore $\solution'$ somehow corresponds to $\solution$.

\begin{lemma}
\label{lem: uncrossing partition}
Suppose that $\Gamma_1$, $\Gamma_2$ are disjoint. Then $\algpop(\Gamma_1,\Gamma_2,\rownaniep)$ is sound and complete.
To be more precise,
if $u = v$ has a non-empty solution \solution{} then for appropriate non-deterministic choices
the returned equation $u' = v'$ has a non-empty solution $\solution'$ such that $\solution'(u') = \sol u$
and $\Gamma_1,\Gamma_2$ is a non-crossing partition with respect to $\solution'$.
\end{lemma}
\begin{proof}
By iterative application of Lemma~\ref{lem:preserving unsatisfiability} we obtain that $\algpop(\Gamma_1,\Gamma_2,\rownaniep)$ is sound.

Concerning the second part of the lemma, for simplicity of presentation we deal only
with the first part of \algpop, i.e.\ the one in which the letters are popped-down,
the second part is dealt with similarly.

Suppose that $\algpop(\Gamma_1,\Gamma_2,\rownaniep)$ always makes the non-deterministic choices according
to $\solution$ (i.e.\ whenever we make a guess about \sol X or \sol x we guess correctly).
Let us a define a new substitution $\solution'$,
the value of $\solution'(X)$ depends on actions performed on $X$ by \algpop:
\begin{itemize}
	\item if $X$ popped up $b$ and $\sol X = bt$
	(which holds, as $\algpop(\Gamma_1,\Gamma_2,\rownaniep)$
	chose according to \solution{} and so the first letter of \sol X is $b$) then $\solution'(X) = t$;
	\item if $X$ did not pop any letter up then $\solution'(X) = \sol X$.
\end{itemize}
Note that $X$ is removed from the equation if and only if $\solution'(X) = \Omega$.

It is easy to verify that indeed in each case the defined $\solution'$ is a solution of the obtained equation $u' = v'$
and $\solution'(u') = \sol u$, as claimed:
when $X$ is not modified, its substitution is the same, if the pops up $b$,
then its solution loose this $b$.
Furthermore, $\solution'$ is non-empty (as if it is empty then we remove the empty context variable).

So suppose that the partition $\Gamma_1,\Gamma_2$ is crossing with respect to $\solution'$,
i.e.\ there exists $a \in \Gamma_1$ and $b \in \Gamma_2$ such that one of \CPref{1}--\CPref{3} holds.
We consider only the case \CPref{1}, in which
$aX$ is a chain subpattern in $u = v$ and $b$ is the first letter of \sol X,
other cases are shown in a similar way.

Consider, whether $X$ popped up a letter.
\begin{description}
	\item[$X$ popped a letter up] In such a case the father of $X$ is labelled with $b' \in \Gamma_2$, a contradiction,
	as the case assumption is that the father is labelled with $a \in \Gamma_1$
	and $\Gamma_1 \cap \Gamma_2 = \emptyset$.
	\item[$X$ did not pop a letter up] Since we consider the non-deterministic choices made according to \solution,
	we know that the first letter of \sol X is outside $\Gamma_2$. And by definition of $\solution'$ we know that $\solution'(X)$
	has the same first letter as \sol X, i.e.\ outside $\Gamma_2$. A contradiction with the case assumption.
\end{description}

Analysis of cases~\CPref{2}--\CPref{3} leads to a contradiction in a similar way, so the argument is skipped,
which ends the proof.
\qedhere
\end{proof}

\subsection{Uncrossing chains}
Suppose that some unary letter $a$ has a crossing chain with respect to a non-empty solution \solution.
Recall that $a$ has a crossing chain if and only if one of the following holds
for some context variables $X, Y$ (or variable $y$)
\begin{enumerate}[(CC1)]
	\item \label{cc 1} $ax$ (or $aX$) is a chain subpattern in $u = v$ and the first letter of \sol x (\sol X, respectively) is $a$;
	\item \label{cc 2} $Xa$ is a chain subpattern in $u = v$ and $a$ is the last letter of $\sol X$;
	\item \label{cc 3} $XY$ (or $Xy$) is a chain subpattern in $u = v$ and $a$ is the last letter of $\sol X$ as well as the first letter of \sol Y (or \sol y).
\end{enumerate}
The first two cases are symmetric while the third is a composition of the first two.
So suppose that the second case holds.
Then we can replace $X$ with $Xa$ throughout the equation $u = v$ (implicitly changing the solution $\sol X = ta\Omega$ to $\sol X = t$)
but it can still happen that $a$ is the last letter of $\sol X$.
So we keep popping down $a$ until the last letter of \sol X is not $a$, in other words we replace $X$ with $Xa^r$,
where $\sol X = ta^r\Omega$ and the last letter of $t$ is not $a$.
Then $a$ and $X$ can no longer satisfy condition \CCref{2}, as $\solution'(X)$ ends with a letter different than $a$.
A symmetric action and analysis apply to \CCref{1},
and \CCref{3} follows by applying the popping down for $X$ and popping up for $Y$ (or $y$).
To simplify the arguments, for a ground term or context $t$ we say that $a^\ell$ is the $a$\emph{-prefix} of $t$ if $t = a^\ell t'$
and the first letter of $t'$ is not $a$ ($t'$ may be empty).
Similarly, for a ground context $t$ we say that $b^r$ is a $b$-suffix of $t$ if $t = t' b^r \Omega$ and the last letter of $t'$ is not $b$
(in particular, $t'$ may be empty).

\begin{algorithm}[H]
  \caption{\algprefsuff$(\Gamma_1,\rownaniep)${} Uncrossing all chains \label{alg:prefix}}
  \begin{algorithmic}[1]
  \For{$X \in \contextvar$ or $x \in \variables$}
		\State let $a$ be the first of \sol X (or \sol x)
		\If{$a \in \Gamma_1$}
			\State guess $\ell \geq 1$ \label{guess ell} \Comment{$a^\ell$ is the $a$-prefix of \sol X or \sol x}
			\State replace each $X$ (or $x$) in $u = v$ by $a^{\ell} X$ (or $a^{\ell} x$)
				\Comment{$\ell$ is stored using $\Ocomp(\ell)$ bits}
			\par \Comment{implicitly change $\sol X = a^{\ell} t$ to $\sol X = t$ (or $\sol x = a^\ell t$ to $\sol x = t$)}
			\If{$\sol X$ is empty} \Comment{Guess}
				\State remove $X$ from $u = v$: replace each $X(t)$ by $t$
			\EndIf
		\EndIf
  \EndFor
  \For{$X \in \contextvar$}
		\State let $b$ be the last letter of \sol X
		\If{$b \in \Gamma_1$}
			\State guess $r \geq 1$ \label{guess r}
				\Comment{$b^r$ is the $b$-suffix of \sol X}
			\State replace each $X$ in $u = v$ by $Xb^r$
				\Comment{$b^r$ is stored in a compressed form}
			\par \Comment{implicitly change $\sol X = tb^r\Omega$ to $\sol X = t$}
			\If{$\sol X$ is empty} \Comment{Guess}
				\State remove $X$ from $u = v$: replace each $X(t)$ by $t$
			\EndIf
		\EndIf
  \EndFor	
  \end{algorithmic}
\end{algorithm}

\begin{lemma}
\label{lem: algpresuff}
$\algprefsuff(\Gamma_1,\rownaniep)$ is sound and complete;
to be more precise, if $u = v$ has a non-empty solution \solution{}
then for appropriate non-deterministic choices the returned equation $u' = v'$ has a solution $\solution'$ such that\
$\solution'(u') = \sol u$ and there are no crossing chains with respect to $\solution'$.
\end{lemma}

The proof is essentially the same as the proof of Lemma~\ref{lem: uncrossing partition} and so it is omitted.

\subsection{Uncrossing father-leaf pairs}
Now it is left to show how to ensure that there is no crossing father-leaf pair $(f,a)$ with $f \in \Gamma_{\geq 1}$ and $a \in \Gamma_0$.
Recall that there is such a pair $(f,a)$ (with respect to a non-empty \solution) if and only if
one of the following holds for some context variable $X$ and variable $y$:
\begin{enumerate}[(CFL 1)]
	\item \label{cfl 1} $f$ with a son $x$ is a subpattern in $u = v$ and $\sol x = a$ \emph{or}
	\item \label{cfl 2} $Xa$ is a subpattern in $u = v$ and the last letter of \sol X is $f$ \emph{or}
	\item \label{cfl 3} $Xy$ is a subpattern in $u = v$, $\sol y = a$ and $f$ is the last letter of \sol X.
\end{enumerate}

The modifications needed to uncross the father-leaf pair are in fact the only new uncrossing operations,
when compared with the recompression technique for strings,
however, they are similar to the one in the case of uncrossing partition $\Gamma_1,\Gamma_2$.
Note that in some sense we even have a partition: $\Gamma_0$ and $\Gamma_{\geq 1}$ and we want to pop-up from $\Gamma_0$ and pop-down from $\Gamma_{\geq 1}$.
The former operation is trivial, but the details of the latter are not, let us present the intuition.
\begin{itemize}
	\item In~\CFLref{1} we \emph{pop up} the letter $a$ from $x$, which in this case means that we replace each $x$ with $a = \sol x$.
	Since $x$ is no longer in the context equation, we can restrict the solution so that it does not assign any value to $x$.
	\item In~\CFLref{2} we \emph{pop down} the letter $f$:
	let $\sol X = sf(t_1,\ldots,t_{i-1},\Omega,t_{i+1},\ldots,t_m)$, where $s$ is a ground context and each $t_i$ is a ground term and $\ar(f) = m$.
	Then we replace each $X$ with $Xf(x_1,x_2,\ldots,x_{i-1},\Omega,x_{i+1},\ldots,x_m)$,
	where $x_1,\ldots,x_{i-1},x_{i+1},\ldots,x_m$ are fresh variables.
	In this way we implicitly modify the solution $\sol X = sf(t_1,t_2,\ldots,t_{i-1},\Omega,t_{i+1},\ldots,t_m)$ to 
	$\solution'(X) = s$ and add $\solution'(x_j) = t_j$ for $j = 1\ldots, i-1,i+1,\ldots,m$.
	If $\solution'(X)$ is empty, we remove $X$ from the equation.
	\item The third case~\CFLref{3} is a combination of~\CFLref{1}--\CFLref{2}, in which we need to down pop from $X$ and pop up from $y$.
\end{itemize}
It is easy to observe that this procedure can be performed on all $f \in \Gamma_{\geq 1}$ and $a \in \Gamma_0$ in parallel,
as presented in the algorithm below; this uncrosses all father-leaf pair $(f,a)$ for $f \in \Gamma_{\geq 1}$ and $a \in \Gamma_0$.

\begin{algorithm}[H]
  \caption{$\alggenpop(\Gamma_{\geq 1},\Gamma_0,\rownaniep)$ \label{alg:genpop}}
  \begin{algorithmic}[1]
	\For{$x \in \variables$}
		\If{$\sol x \in \Gamma_0$} \Comment{Guess} 
			\State replace each $x$ in $u = v$ by \sol x \label{popup} \Comment{\solution{} is no longer defined on $x$}
		\EndIf
	\EndFor
	\For{$X \in \contextvar$}
		\State let $f$ be the last letter of \sol X \Comment{Guess}
		\If{$f \in \Gamma_{\geq 1}$ and for some $a\in \Gamma_0$ the $Xa$ is a subpattern in $u = v$} \label{applied on constant}
			\State let $m = \ar(f)$
			\State let $i$ be such that $\Omega$ labels the $i$-th child of its father in \sol X \Comment{Guess}
			\State replace each $X$ in $u = v$ by $Xf(x_1,x_2,\ldots, x_{i-1},\Omega,x_{i+1},\ldots,x_m)$ \label{popdown}
			\par
			\Comment{Implicitly change $\sol X = s f(t_1,t_2,\ldots, t_{i-1},\Omega,t_{i+1},\ldots,t_m)$ to $\sol X = s$}
			\par \Comment{Add new variables $x_1,\ldots, x_m$ to \variables{} with $\sol {x_j} = t_j$}
			\If{$\sol X$ is empty} \Comment{Guess}
				\State remove $X$ from the equation: replace each $X(u)$ by $u$
			\EndIf
		\EndIf
	\EndFor
	\For{new variables $x \in \variables$}
		\If{$\sol x \in \Gamma_0$} \Comment{Guess} 
			\State replace each $x$ in $u = v$ by \sol x \label{popup2} \Comment{\solution{} is no longer defined on $x$}
		\EndIf
	\EndFor
  \end{algorithmic}
\end{algorithm}

There is a subtle difference between uncrossing a partition $\Gamma_1,\Gamma_2$ and uncrossing father-leaf pairs:
for a partition popping down letters from $\Gamma_1$ is unconditional
while the corresponding popping down the last letters $f \in \Gamma_{\geq 1}$ from $X$
is done only when it is really needed:
i.e.\ we want to make some $(f,i,a)$ leaf compression,
$f$ is the last letter of $\sol X$, its $i$-th child is $\Omega$ and some occurrence of $X$ is applied on $a$.
This assumption turns out to be crucial to bound the number of introduced variables, see Lemma~\ref{lem: owned variables}.

\begin{lemma}
\label{lem: uncrossing father leaf}
Let $\Gamma_{\geq 1}$ be a set of some letters of arity at least $1$ and $\Gamma_0$ set of some constants,
then $\alggenpop(\Gamma_{\geq 1},\Gamma_0,\rownaniep)$ is sound.

It is complete, to be more precise,
if $u = v$ has a non-empty solution \solution{} then for appropriate non-deterministic choices
the returned equation $u' = v'$ has a non-empty solution $\solution'$ such that $\solution'(u') = \sol u$
and there is no crossing father-leaf pair $(f,a)$ with $f \in \Gamma_{\geq 1}$ and $a \in \Gamma_0$ with respect to $\solution'$.
\end{lemma}
\begin{proof}
The proof is similar as in the case of Lemma~\ref{lem: uncrossing partition}, however, some details are different so it is supplied.

By iterative application of Lemma~\ref{lem:preserving unsatisfiability}
we obtain that $\alggenpop(\Gamma_{\geq 1},\Gamma_0,\rownaniep)$ is sound.

Concerning the second part of the lemma, we proceed as in Lemma~\ref{lem:preserving unsatisfiability}:
let $\alggenpop(\Gamma_{\geq 1},\Gamma_0,\rownaniep)$ always make the non-deterministic choices according
to the $\solution$: we replace $x$ with $a$ when $\sol x = a \in \Gamma_0$
and when we pop down $f(x_1,\ldots,x_{i-1},\Omega,x_{i+1},\ldots,x_m)$ from $X$ then indeed $f$ is the last letter
of \sol X and $\Omega$ labels the $i$-th child of $f$.
We define a new substitution $\solution'$:
\begin{itemize}
	\item The values on old variables do not change, i.e.\ $\solution'(x) = \sol x$ for each variable $x$ present in the context equation
both before and after \alggenpop.
	\item For a context variable $X$ from which we did not pop a letter we set $\solution'(X) = \sol X$.
	\item For $X$ from which \alggenpop{} popped down $f(x_1,\ldots,x_{i-1},\Omega,x_{i+1},\ldots,x_m)$
let $\sol X = s f(t_1,\ldots,t_{i-1},\Omega,t_{i+1},\ldots,t_m)$
(such a representation is possible as \alggenpop{} guesses according to \solution).
Then we define $\solution'(X) = s$ and $\solution'(x_j) = t_j$ for $j = 1,\ldots, i-1,i+1,\ldots, m$.
Note that when $s = \Omega$ then $X$ is removed from the equation.
	\item For $x$ that popped-up a constant we do not need to define \sol x as it is no longer in the context equation.
\end{itemize}

It is easy to verify that indeed in each case the defined $\solution'$ is a solution of the obtained equation $u' = v'$
and $\solution'(u') = \sol u$, as claimed.

So suppose that there is a crossing father-leaf pair $(f,a)$ with $f \in \Gamma_{\geq 1}$ and $a \in \Gamma_0$ with respect to $\solution'$,
i.e.\ one of the \CFLref{1}--\CFLref{3} holds.
Note that in \CFLref{1} and \CFLref{3} there is a variable $y$ such that $\solution'(y) \in \Gamma_0$,
however, by our assumption that \alggenpop{} always makes the choice according to the \solution{}
each such variable $y$ was replaced with \sol y in the context equation in line~\ref{popup} or line~\ref{popup2}.
So it remains to consider the \CFLref{2}.

So let $X$ be as in \CFLref{2}, i.e.\ the last letter of $\solution'(X)$ is $f \in \Gamma_{\geq 1}$ and $Xa$ is a subpattern in $u = v$
for some $a \in \Gamma_0$.
Consider, whether $X$ popped down a letter:
\begin{description}
	\item[$X$ popped a letter down] Then for each occurrence of subpattern $Xt$ in the context equation,
	the first letter of $t$ is always some $g \in \Gamma_{\geq 1}$ (as there was no way to change this),
	since $\Gamma_0$ and $\Gamma_{\geq 1}$ are disjoint, this is a contradiction with the assumption that $Xa$ is a subpattern in the equation for some $a \in \Gamma_0$.
	\item[$X$ did not pop a letter down]
	Consider the occurrence of a subpattern $Xa$.
	This $a$ letter was there when we decided not to pop down a letter from $X$ in line~\ref{applied on constant}.
	Then $\algpop(\Gamma_{\geq 1},\Gamma_0,\rownaniep)$ should have popped the last letter of $f$ from $X$,
	as in line~\ref{applied on constant} we were supposed to guess according to \solution, contradiction. \qedhere
\end{description}
\end{proof}

\section{Main algorithm}
\label{sec: main algorithm}
Now we are ready to describe the whole algorithm for testing the satisfiability of context equations.
It works in \emph{phases}, each of which is divided into two subphases.
In each subphase we first perform the chain compression, the $\Gamma_1$, $\Gamma_2$ compression
for appropriate partition $\Gamma_1$, $\Gamma_2$ and lastly the leaf compression.
In order to make the chain compression we first uncross all chains, similarly in order to perform the $\Gamma_1$, $\Gamma_2$ compression
we ensure that $\Gamma_1$, $\Gamma_2$ is a non-crossing partition
and in order to make the leaf compression we make sure that there is no crossing father-leaf pair.

The reason to have two subphases is quite simple:
(for appropriate guess of partition) the first subphase ensures that the size of the (size-minimal) solution
decreases by a constant factor (cf.~Theorem~\ref{thm: tree size drop}),
the second phase is used to make sure that the size of the equation is bounded
(in some sense the second phase decreases the size of the equation,
but as the equation grows in the first subphase, in total we can only guarantee that the equation is of more or less the same size).

\begin{algorithm}[H]
	\caption{$\algsolveeq(\rownaniep,\Sigma)$ Checking the satisfiability of a context equation $u = v$ over signature $\Sigma$}
	\label{alg:main}
	\begin{algorithmic}[1]
	\State let $k \gets$ maximal arity of functions from $\Sigma$
	\While{$|u| > 1$ or $|v| > 1$} \label{alg:mainloop}
		\For{$i\gets 1 \twodots 2$} \label{two iterations} \Comment{One iteration to shorten the solution, one to shorten the equation}
			\State $\Gamma_1 \gets$ unary letters in $\rownanie$ \Comment{By Lemma~\ref{lem: almost over a cut}}
			\State $\algprefsuff(\Gamma_1,\rownaniep)$ \Comment{No letter has a crossing block}
			\State $\algchaincompncr(\Gamma_1,\rownaniep)$	\Comment{Chain compression}
			\State $\Gamma \gets $ the set of unary in $u = v$ \Comment{By Lemma~\ref{lem: almost over a cut}}
			\State guess partition of $\Gamma$ into $\Gamma_1$ and $\Gamma_2$
			\State $\algpop(\Gamma_1,\Gamma_2,`u=v')$ \Comment{$\Gamma_1$, $\Gamma_2$ is a non-crossing partition} 
			\State $\algpaircompncr(\Gamma_1,\Gamma_2,\rownaniep)$ \Comment{$\Gamma_1$, $\Gamma_2$ compression}
			\State $\Gamma_{\geq 1} \gets $ non-constants in `$u = v$' plus one fresh letter $f_i$ of arity $i$ for each $1 < i \leq k$ \par
			\Comment{By Lemma~\ref{lem: almost over a cut}}
			\State $\Gamma_0 \gets $ constants in `$u = v$' plus one fresh constant $c$ \Comment{By Lemma~\ref{lem: almost over a cut}}
			\State $\alggenpop(\Gamma_{\geq 1},\Gamma_0,\rownaniep)$ \Comment{No crossing father-leaf pairs}
			\State $\algchildcompncr(\Gamma_{\geq 1},\Gamma_0,\rownaniep)$
		\EndFor
	\EndWhile
	\State Solve the problem naively \label{naive solve}
		\Comment{With sides of size $1$, the problem is trivial}
 \end{algorithmic}
\end{algorithm}

The properties of \algsolveeq{} are summarised in the following lemma
\begin{theorem}
\label{thm: main}
\algsolveeq{} stores equation of length $\Ocomp(nk)$ and uses additional $\Ocomp(n^2k^2)$ memory,
where $n$ is the size of the input equation while $k$ is the maximal arity of symbols from $\Sigma$.
It non-deterministically solves context equation, in the sense that:
\begin{itemize}
	\item if the input equation is not-satisfiable then it returns `NO';
	\item if the input equation is satisfiable then for some nondeterministic choices in $\Ocomp(\log N)$ phases it returns `YES',
	where $N$ is the size of size-minimal solution.
\end{itemize}
\end{theorem}

As a corollary we get an upper bound on the computational complexity of context unification.
\begin{corollary}
Context unification is in \PSPACE.
\end{corollary}
\begin{proof}
By Theorem~\ref{thm: main} the (non-deterministic) algorithm \algsolveeq{} works in space $\Ocomp(n^2k^2)$,
which is polynomial in the input size.
By Savitch Theorem the non-deterministic polynomial space algorithm can be determinised, using at most quadratically more space.
\qedhere
\end{proof}

\subsection{Analysis}
The actual statement needed to show Theorem~\ref{thm: main} is given in the below technical lemma.

\begin{lemma}
	\label{lem: technical}
\algsolveeq{} is sound.

It is complete, to be more precise for some nondeterministic choices the following conditions are satisfied:
\begin{enumerate}
	\item the stored context equation has size $\Ocomp(nk)$, with at most $n$ context variables and $kn$ variables;
	\item if $N$ is the size of the size-minimal solution at the beginning of the phase then at the end of the phase the
	equation has a solution of size at most $\frac{3N}{4}$;
	\item the additional memory usage is at most $\Ocomp(k^2n^2)$ (counted in bits);
	\item the maximal arity of symbols in $\Sigma$ does not increase during \algsolveeq.
\end{enumerate}
\end{lemma}
The rest of this subsection is devoted to the proof of Lemma~\ref{lem: technical}.

\subsubsection{Number of phases}
We show that the number of phases is logarithmic in $N$:
we show that one subphase of \algsolveeq{} in some sense can simulate an action of \algtreecomp{}
on a size-minimal solution of an equation.
Thus, by Theorem~\ref{thm: tree size drop} the size of the length-minimal solution drops by a constant in a phase.
Due to the presence of letters that are not in the equation, we cannot guarantee that this solution prevails the compression
steps, however, the size of the length-minimal solution does drop by a constant factor.

\begin{lemma}
\label{lem: number of phases}
Let the size-minimal solution of $u = v$ has size $N$.
Then for appropriate non-deterministic choices after first subphase of \algsolveeq{} the obtained equation $u' = v'$
has a solution $\solution'$ of size $N' \leq \frac{3N}{4}$.
\end{lemma}
\begin{proof}
Consider some size-minimal solution $\solution_{\min}$ of size $N$
and let $\Gamma_1$ be the set of unary letters in \rownanie.
By Lemma~\ref{lem: almost over a cut} there is another solution \solution{} of the same size $N$
that uses only unary letters from $\Gamma_1$: by the first part of Lemma~\ref{lem: almost over a cut}
we can find a solution of the same size with at most one unary letter not used in $u = v$
and then the second part of the Lemma guarantees that we can make the solution even smaller by deleting all occurrences of this letter
(which contradicts the size-minimality of the solution).

By Lemma~\ref{lem: algpresuff} for appropriate non-deterministic choices after the $\algprefsuff(\Gamma_1,\rownaniep)$
new equation $u_1 = v_1$ has a solution $\solution_1$ such that $\solution_1(u_1) = \sol u$
and there are no crossing chains for $a \in \Gamma_1$ with respect to $\solution_1$.
Then by Lemma~\ref{lem: algchaincomp noncrossing} after the $\algchaincompncr(\Gamma_1,\rownaniep[1])$
the obtained equation $u_2 = u_2$ has a solution $\solution_2$
such that $\solution_2(u_2) = \algtreechaincomp(\Gamma_1,\solution_1(u_1))$.
Note that clearly $|\solution_2(u_2)| \leq N$.

Consider $\solution_2$.
In a similar way as for $\solution_{\text{min}}$ we can show using
Lemma~\ref{lem: almost over a cut} that there is a solution $\solution_2'$ such that $\solution_2'(u_2) \leq \solution_2(u_2)$
such that $\solution_2'$ uses only unary letters that are used in \rownanie[2].
By Lemma~\ref{lem: tree size drop technical} there is some partition of unary letters used in `$u_2 = v_2$' into $\Gamma_1$ and $\Gamma_2$ such that
the $\Gamma_1$, $\Gamma_2$ compression followed by the leaf-compression results in a tree of size at most $\frac{3}{4} |\solution_2'(u_2)|$,
which is at most $\frac{3}{4} N$; fix this partition $\Gamma_1$, $\Gamma_2$ for the remainder of the proof.

We perform $\algpop(\Gamma_1,\Gamma_2,\rownaniep[2])$, by Lemma~\ref{lem: uncrossing partition} for appropriate non-deterministic choices the returned equation
$u_3 = v_3$ has a solution $\solution_3$ such that $\solution_3(u_3) = \solution_2'(u_2)$ and $\Gamma_1,\Gamma_2$ is a non-crossing partition
with respect to $\solution_3$.

We apply $\algpaircompncr(\Gamma_1,\Gamma_2,$`$u_3 = v_3$'$)$, since the partition $\Gamma_1$, $\Gamma_2$ is non-crossing for $u_3 = v_3$
with respect to $\solution_3$, by Lemma~\ref{lem: algpaircomp noncrossing} the obtained equation $u_4 = v_4$
has a solution $\solution_4$ such that
$\solution_4(u_4) = \algtreepaircomp(\Gamma_1,\Gamma_2,\solution_3(u_3))$.

Finally, consider the solution $\solution_4$ of $u_4  = v_4$.
By Lemma~\ref{lem: almost over a cut} there is a solution $\solution_4'$ that is a simpler equivalent
(in particular, $\solution_4'(u_4)$ has the same number of constants as $\solution_4(u_4)$)
and uses only one letter per arity that is not used by `$u_4 = v_4$'.
Let $\Gamma_{\geq 1}'$ denote the set of letters of arity greater than $1$ used in $\solution_4'(u_4)$ and $\Gamma_{\geq 1}$ in $\solution_4(u_4)$
while $\Gamma_0'$ be the set of constants used in $\solution_4'(u_4)$ and $\Gamma_0$ in $\solution_4(u_4)$.
Observe that $\algtreechildcomp(\Gamma_{\geq 1},\Gamma_0,\solution_4(u_4))$ and $\algtreechildcomp(\Gamma_{\geq 1}',\Gamma_0',\solution_4'(u_4))$
have the same size, as $\solution_4'$ is a simpler equivalent of $\solution_4$ implies that
$\solution_4'(u_4)$ has the same number of constants as $\solution_4(u_4)$
and \algtreechildcomp{} on both of them simply compresses all leaves to their respective fathers.
Hence, by Lemma~\ref{lem: tree size drop technical},
$\algtreechildcomp(\Gamma_{\geq 1}',\Gamma_0',\solution_4'(u_4))$ has size at most
$\frac{3}{4}|\solution_2'(u_2)| \leq \frac{3}{4}|\solution(u)| = \frac{3}{4} N$

So it is left to show that we can simulate $\algtreechildcomp(\Gamma_{\geq 1}',\Gamma_0',\solution_4'(u_4))$ on the equation.
So consider $\solution_4'$, $\Gamma_{\geq 1}'$ and $\Gamma_0'$.
By Lemma~\ref{lem: uncrossing father leaf} for appropriate non-deterministic choices
after $\alggenpop(\Gamma_{\geq 1}',\Gamma_0',\rownaniep[4])$ the obtained equation \rownanie[5] has a solution $\solution_5$ such that
$\solution_5(u_5) = \solution_4'(u_4)$ and there is no crossing father-leaf pair $(f,a)$ with $f \in \Gamma_1'$ and $a \in \Gamma_0'$
with respect to $\solution_5$.
We now apply $\algchildcompncr(\Gamma_{\geq 1}',\Gamma_0',\rownaniep[5])$.
By Lemma~\ref{lem: algchildcomp noncrossing} for appropriate non-deterministic choices the returned equation $\rownanie[6]$
has a solution $\solution_6$ such that $\solution_6(u_6) = \algtreechildcomp(\Gamma_{\geq 1}',\Gamma_0',\solution_5(u_5)) = \algtreechildcomp(\Gamma_{\geq 1}',\Gamma_0',\solution_4'(u_4))$.
In particular, this solution is as small as promised in the lemma.
\qedhere
\end{proof}

\subsubsection{Space consumption}
Observe that, in contrast to the recompression-based algorithm for word equations,
\algsolveeq{} introduces new variables and their occurrences to the equation (when \alggenpop{} pops down a letter of arity greater than $1$).
At first it seems like a large problem, as the number of letters introduced to the equation in one phase depends on the number of variables,
however, we are able of bounding the number of such variables at any given time of the \algsolveeq{}
by $kn$.
To this end, we need some definitions: we say that a variable $x_i$ is \emph{owned} by a context variable $X$
if $x_i$ occurred in the equation when $X$ popped a letter down.
A particular occurrence of $x_i$ in the equation is \emph{owned} by the occurrence of the context variable that introduced it.
When a context variable $X$ is removed from the equation the variables its owns get \emph{disowned} (and particular occurrences
of this variable are also disowned).

We show that each context variable owns at most $k-1$ variables.

\begin{lemma}
\label{lem: owned variables}
Every context variable $X$ present in $u = v$ owns at most $k-1$ variables.
In particular, there are at most $kn$ occurrences of variables in $u = v$.
\end{lemma}
Note that the upper bound on the number of variables \emph{does not} depend on the non-deterministic choices of \algsolveeq.
\begin{proof}
Given an occurrence of a subterm $Xt$ we say that this occurrence of $X$ dominates the occurrences of variables in $t$.

We show by induction two technical claims:
\begin{enumerate}
		\item For every occurrence of a variable $X$ the multiset of variables, whose occurrences it owns, is the same.
		\label{same owned variables}
		\item Each appearance of $X$ dominates its owned occurrences of variables. \label{chains to owned variables}
\end{enumerate}

The subclaim~\ref{same owned variables} is trivial:
at the beginning, there are no owned variables.
When we introduce new $X$-owned variables,
we replace each $X$ with the same $Xf(x_1,\ldots,x_{i-1},\Omega,x_{i+1},\ldots,x_m)$,
in particular the set of $X$-owned variables for each occurrence of $X$ is increased by $\{x_1,\ldots,x_{i-1},x_{i+1},\ldots,x_m\}$.
When we remove occurrences of $x$, we remove them all at the same time.
Which ends the induction.

Concerning the subclaim~\ref{chains to owned variables},
this vacuously holds for the input instance, which yields the induction base.
For the induction step, consider now the operation performed by \algsolveeq{} on the context equation.
Any compression is performed only on letters, so it cannot affect the domination.
When we pop the letters from a variable $x$, we replace $x$ with $ax$ (or remove $x$ altogether),
so this also does not affect the domination.
Similarly, when we pop letters from context variables, we either replace $X$ with $aX$ or $X$ with
$Xf(x_1,\ldots,x_{i-1},\Omega,x_{i+1},\ldots,x_m)$,
in both cases the domination of the old variables is not affected and in the last case the new variables
$x_1,\ldots, x_m$ owned by this particular occurrence of $X$ are indeed dominated by this occurrence of $X$.

Using those two subclaims we now show that if during \alggenpop{} $X$ pops down a letter, then $X$ does not own any variables.
Suppose that $X$ pops down a letter.
Then in \rownanie{} there is a subtree $Xc$ for $c \in \Gamma_0$.
Suppose that $X$ owned a variable $x$ before popping down the letter.
Then by subclaim~\ref{same owned variables} the occurrence which is applied on $c$
also owns occurrence of $x$ and by~\ref{chains to owned variables} this occurrence is dominated by its owning occurrence of $X$,
which is not possible, as this owning occurrence of $X$ is part of the term $Xc$.
As a consequence, each occurrence of a context variable owns at most $k-1$ occurrences of variables.

Now, concerning the number of variables: initially there are at most (not owned nor disowned) $n$ variables occurrences and $n$ context variables occurrences.
Suppose that at some point there are $m \leq n$ context variables occurrences.
Since we never introduce context variables, there are at most $m(k-1)$ owned variables' occurrences,
and at most $(n - m)(k - 1)$ disowned ones and $n$ that are neither owned, nor disowned
(those are the occurrences of variables that were present in the input equation),
so $nk$ occurrences of variables in total, as claimed.
\end{proof}

We move to the crucial part of the proof: the space consumption of \algsolveeq.
The intuition should be clear: in the second subphase we treat the equation as a term
and try to ensure that its size drops by one fourth, just as in the case of Theorem~\ref{thm: tree size drop}.
However, in the meantime we also increased the size of the equation, as we pop the letters into the context equation (in both subphases).
The number of those letters depends linearly on the number of occurrences of variables and context variables in \rownanie,
which is known to be $\Ocomp(kn)$, see Lemma~\ref{lem: owned variables}.
Hence those two effects (increasing the size and reducing the size)
cancel each out and it can be shown that the size of the equation is $\Ocomp(kn)$.

\begin{lemma}
\label{lem: equation size}
For appropriate non-deterministic choices in second subphase
the context equation at the end of a phase of \algsolveeq{} has size $\Ocomp(nk)$.
Furthermore, for those choices \algsolveeq{} is complete, to be more precise:
If $u = v$ after the first subphase had a solution \solution{} of size $N$ 
then after the second subphase the obtained equation $u' = v'$ has a solution of size at most $N$.
\end{lemma}
\begin{proof}
We show some nondeterministic choices for which the run of \algsolveeq{} in the second subphase satisfies the claim of the lemma.

In the following we consider only the number of letters in \rownanie: since no new context variables are introduced,
there are at most $n$ such occurrences
and by Lemma~\ref{lem: owned variables} there are at most $kn$ occurrences of variables in \rownanie.

Consider, how many new letters were introduced during the first subphase into the equation.

\begin{itemize}
	\item The chains introduced by \algprefsuff{} are immediately replaced with a single letter,
therefore we can think that \algprefsuff{} introduces $2$ letters per context variable and $1$ per variable,
so at most $2n + kn$ in total.
(Note that popping a letter down introduces also variables, but those are counted separately).
	\item 
Similarly, each \algpop{} introduces at most $2$ letters per context variables and $1$ per variable,
so also at most $2n + kn$ in total.
	\item 
Lastly, the popping down in \alggenpop{} introduces $1$ letter per context variable, so $n$ letters
while the popping up introduces at most $1$ letter per variable,
but all those letter are compressed into their parents immediately afterwards (in \algchildcompncr:
all letters popped up are from $\Gamma_0$ and by choice of $\Gamma_{\geq 1}$ their fathers in \sol U or \sol v
are from $\Gamma_1$, so they are compressed.
Thus we do not need to count them.
\end{itemize}
Hence, in total, during the first subphase the size of the equation increases by at most $5n +2kn$ letters.

Concerning the second subphase,
the following analysis is similar to the one in Lemma~\ref{lem: tree size drop technical} but it takes into the account also the letters introduced
due to popping.
Furthermore, we need to also guarantee that the equation stays satisfiable and the the size of the size-minimal solution does not increase.

As in Lemma~\ref{lem: tree size drop technical},
let $n_0$, $n_1$ and $n_{\geq 2}$ denote the number of letters of arity $0$, $1$ and at least $2$, respectively
in the equation \rownanie{} between the first and second subphase;
$n_0'$, $n_1'$ and $n_{\geq 2}'$ the number of letters of arity $0$, $1$ and at least $2$
in \rownanie{} after the chain compression, $n_0''$, $n_1''$ and $n_{\geq 2}''$ after the (appropriate) $\Gamma_1$, $\Gamma_2$ compression
and finally $n_0'''$, $n_1'''$ and $n_{\geq 2}'''$ after the leaf compression.
We shall show that
\begin{equation}
\label{eq: end of a phase}
n_0''' + n_1''' + n_{\geq 2}''' \leq \frac{3}{4}\Big(n_0 + n_1 + n_{\geq 2} \Big) + f(n,k) \enspace,
\end{equation}
where $f$ is some function linear in $n$ and $k$.
Taking into the account that during the first subphase the size of the equation increased by at most $5n + 2kn$, we obtain that
the equation at the end of the phase is of size at most 
(let $m$ be the size of the equation at the beginning of the phase)
$$
\frac{3}{4} m + \frac{15}{4}n + \frac{6}{4}kn + f(n,k) \enspace.
$$
From which by an easy induction it follows that
$$
m \leq 15 n + 6kn + 4f(n,k) \enspace .
$$
So it is left to show that estimation~\eqref{eq: end of a phase} indeed holds.

\begin{subequations}
Let $u = v$ be a satisfiable equation after the first subphase,
let $\solution$ be one of its size-minimal equation
(note that in this proof in general we do not need to worry 
about the letters that are not present in the equation, as we focus on compressing letters in the equation).
Let $\Gamma_1$ be the set of letters present in the equation $u = v$.
After the $\algprefsuff(\Gamma_1,\rownaniep)$ the obtained equation $u_1 = v_1$
has a solution $\solution_1$ such that $\solution_1(u_1) = \solution(u)$.
In particular
the size-minimal solution of $u_1 = v_1$ is not larger than the one of $u = v$.

Note that each chain popped from a context variable or variable by $\algprefsuff(\Gamma_1,\rownaniep)$ is immediately replaced with 
a single letter during the $\algchaincompncr(\Gamma_1,\rownaniep[1])$,
there are at most $2n + kn$ new unary) letters introduced to \rownanie, hence
\begin{equation}
\label{eq: after blocks}
n_0' = n_0 \quad n_1' \leq n_1 + 2n + kn \quad n_{\geq 2}' = n_{\geq 2} \enspace .
\end{equation}
Furthermore, by Lemma~\ref{lem: algchaincomp noncrossing} the obtained equation $u_2 = v_2$ has a solution $\solution_2$
such that $\solution_2(u_2) = \algtreechaincomp(\Gamma_1,\solution_1'(u_1))$,
in particular $|\solution_2(u_2)| \leq |\algtreechaincomp(\Gamma_1,\solution_1'(u_1))| \leq N$.

Now consider the maximal chains in \rownanie{} that are formed only by letters from $\Sigma$,
i.e.\ a context variable denotes the beginning or the end of such a chain, let $c$ denotes the number of such chains.
Then
\begin{equation}
	\label{eq: number of chains}
c \leq\frac{n_0}{2} + n_{\geq 2} + n + \frac{kn}{2} + \frac{1}{2} \enspace .
\end{equation}
Indeed: consider any such chain and the bottom unary letter in it.
Then its child is labelled with either a letter, a context variable or a variable.
Similarly, consider the topmost letter in any such chain, then the father is labelled with either a letter of arity at least $2$ or by
a context variable or nothing at all, when this node is a root.
(Note that without loss of generality we may assume that in \rownanie{} at most one root is labelled with a unary letter: if both
$u= a u'$ and $v= av'$ then we can simply replace \rownanie{} with `$u' = v'$', if their first letters are different then the equation
is trivially not satisfiable.)
Summing up those two estimations we get that
\begin{align*}
2c
	&\leq
\underbrace{n_0' + n_{\geq 2}' + kn + n}_{\text{nodes below}} + \underbrace{n_{\geq 2}' + n}_{\text{nodes above}} + \underbrace{1}_{\text{possible root}}\\
	&\leq
n_0 + 2n_{\geq 2} + 2n + kn + 1 \enspace ,
\end{align*}
which yields~\eqref{eq: number of chains}.

Let $\Gamma$ denote the set of unary letters in \rownanie.
Then there is a partition of $\Gamma$ into $\Gamma_1$ and $\Gamma_2$ such that at least $\frac{n_1' - c}{4}$ pairs in $u = v$
are covered by a partition $\Gamma_1$, $\Gamma_2$:
this follows by a randomised argument similar to the one in Claim~\ref{clm: random partition}.
Fix this partition for the remainder of the proof.

We first perform the $\algpop(\Gamma_1,\Gamma_2,\rownaniep[2])$ (obtaining $u_3 = v_3$)
and then $\algpaircompncr(\Gamma_1,\Gamma_2,\rownaniep[3])$.
The former operation introduces at most $2n + kn$ unary letters to the equation,
while the latter compresses at least $\frac{n_1' - c}{4}$ unary letters.
Hence
\begin{equation}
	\label{eq: second compression}
n_0'' = n_1' = n_0 \quad n_1'' \leq \frac{3}{4}n_1' + \frac{c}{4} + 2n + kn \quad n_{\geq 2}'' = n_{\geq 2}' = n_{\geq 2} \enspace .
\end{equation}
Let us elaborate on the estimation for $n_1''$:
\begin{align}
\notag
n_1''
	&\leq
\frac{3}{4}n_1' + \frac{c}{4} + 2n + kn &\text{from~\eqref{eq: second compression}}\\
	\notag
	&\leq
\underbrace{\frac{3n_1}{4} + \frac{3n}{2} + \frac{3kn}{4}}_{3/4 n_1'}
	+
\underbrace{\frac{n_0}{8} + \frac{n_{\geq 2}}{4} + \frac{n}{4} + \frac{kn}{8} + \frac{1}{8}}_{c/4} + 2n + kn
&\text{from~\eqref{eq: after blocks} and~\eqref{eq: number of chains}}\\
\label{eq: n1''}
	&=
\frac{n_0}{8} + \frac{3n_1}{4} + \frac{n_{\geq 2}}{4} + \frac{15n}{4} + \frac{15kn}{8} + \frac{1}{8} &\text{simplification}\enspace .
\end{align}
Observe that after the $\algpop(\Gamma_1,\Gamma_2,\rownaniep[2])$,
by Lemma~\ref{lem: uncrossing partition} the obtained equation $u_3 = v_3$
has a solution $\solution_3$ such that $\solution_3(u_3) = \solution_2(u_2)$ (so also $|\solution_3(u_3)| \leq N$)
and $\Gamma_1,\Gamma_2$ is a non-crossing partition with respect to $\solution_3$.
Then by Lemma~\ref{lem: algpaircomp noncrossing} the following $\algpaircompncr(\Gamma_1,\Gamma_2,\rownaniep[3])$ returns an equation $u_4 = v_4$
which has a solution $\solution_4$ such that $\solution_4(u_4) = \algtreepaircomp(\Gamma_1,\Gamma_2,\solution_3(u_3))$,
in particular $|\solution_4(u_4)| \leq N$.
By Lemma~\ref{lem: almost over a cut}, there is also a solution $\solution_4'$ of $u_4 = v_4$ that uses at most
one letter of each arity that is not in $u_4 = v_4$.

Now, lastly, during the leaf compression we first pop up letters
(i.e.\ we replace some variables by constants) then we pop letters down,
introducing one letter per context variables, so at most $n$ letters,
that are of arity at least $1$ and then again pop up.
As $\solution_4'$ uses at most one letter of each arity that is not in $u_4 = v_4$
we may assume that \algsolveeq{} guesses those letters into $\Gamma_{\geq 1}$ and $\Gamma_0$.
Hence the letters that are popped-up (i.e.\ they replace some variables) are from $\Gamma_0$
and are immediately compressed to their fathers (who are from $\Gamma_1$) during the leaf compression,
so we may ignore the letters that are popped up for the purpose of our estimation.
On the other hand, each leaf labelled with a letter is also compressed, i.e.\ $n_{\geq 2}''$ letters are compressed.
Hence
\begin{align*}
n_0''' &+ n_1''' + n_{\geq 2}'''\\
	\quad &\leq
n_0'' + n_1'' + n_{\geq 2}'' + n - n_0'' &\text{popped up and absorbed letters}\\
	&=
n_1'' + n_{\geq 2}'' + n &\text{simplification}\\
	&\leq
\underbrace{\frac{n_0}{8} + \frac{3n_1}{4} + \frac{n_{\geq 2}}{4} + \frac{15n}{4} + \frac{15kn}{8} + \frac{1}{8}}_{n_1''} + n_{\geq 2} + n &\text{from~\eqref{eq: n1''}}\\
	&=
\frac{n_0}{8} + \frac{3n_1}{4} + \frac{5n_{\geq 2}}{4} + \frac{19n}{4} + \frac{15kn}{8} + \frac{1}{8} &\text{simplification}\\
	&\leq
\frac{n_0}{8} + \frac{3n_1}{4} + \underbrace{\frac{3n_{\geq 2}}{4} + \frac{1}{2}\underbrace{(n_0 + kn-1)}_{\geq n_{\geq 2}}}_{\geq \frac{5}{4}n_{\geq 2}} + \frac{19n}{4} + \frac{15kn}{8} + \frac{1}{8} &\text{from~\eqref{eq: after blocks}}\\
	&<
\frac{3}{4}\Big(n_0 + n_1 + \geq 2\Big) + \frac{19n}{4} + \frac{19kn}{8} &\text{simplification} \enspace.
\end{align*}
Which shows~\eqref{eq: end of a phase} for $f(n,k) = \frac{19n}{4} + \frac{19kn}{8}$ and so ends the proof.
Concerning the satisfiability,
by Lemma~\ref{lem: uncrossing father leaf} after the $\alggenpop(\Gamma_{\geq 1},\Gamma_0,\rownaniep[4])$
the obtained equation $u_5 = v_5$ has a solution $\solution_5$ such that
$\solution_5(u_5) = \solution_4'(u_4)$, hence also $|\solution_5(u_5)| \leq N$, and there is no crossing father-leaf pair $(f,a)$
with $f \in \Gamma_{\geq 1}$ and $a \in \Gamma_0$.
Then, by Lemma~\ref{lem: algchildcomp noncrossing},
the following $\algchildcompncr(\Gamma_{\geq 1},\Gamma_0,\rownaniep[5])$ returns an equation $u_6 = v_6$ with a solution $\solution_6$
such that $\solution_6(u_6) = \algtreechildcomp(\Gamma_{\geq 1},\Gamma_0,\solution_5(u_5))$,
hence $u_6 = v_6$ is satisfiable and it has a solution of size at most $N$.
\end{subequations}
\qedhere
\end{proof}

Now showing Lemma~\ref{lem: technical} follows naturally.

\begin{proof}[proof of Lemma~\ref{lem: technical}]

\begin{itemize}
	\item The bound on number of occurrences of variables follows from Lemma~\ref{lem: owned variables}.
	No context variables are introduced, so there are at most $n$ occurrences of context-variables.
	The bound on the size of the equation follows from Lemma~\ref{lem: equation size}.
	\item The bound on the size of the size-minimal solution after one phase follows from Lemma~\ref{lem: number of phases}
	and Lemma~\ref{lem: equation size}: by the former it is reduced by a factor of $1/4$ during the first subphase
	and by the latter the size of the size-minimal solution does not increase during the second subphase.
	
	Concerning the additional memory usage: storing an equation of length $\Ocomp(nk)$ uses $\Ocomp(nk \log(nk))$ bits.
	Additionally, we need to store the lengths of the popped chains of letters (we store $a^\ell$ as a pair $(a,\ell)$).
	Without loss of generality we can focus on size-minimal solutions,
	for whose those lengths are of size $2^{c(|u| + |v|)}$ for some constant $c$, by Lemma~\ref{lem: exponent of periodicity},
	so each can be encoded using $\Ocomp(|u| + |v|) = \Ocomp(kn)$ bits; there are at most $n$ context variables and $kn$ variables
	(by Lemma~\ref{lem: owned variables}),
	so there are $\Ocomp(kn)$ such prefixes and suffixes, so in total we need $\Ocomp(k^2n^2)$ bits to denote them.
	All other operations increase the space usage by a constant factor only.
	\item The bound on the arity of letters is an easy observation, similar to the one in Lemma~\ref{lem: keeping arity low}:
	no operation introduces letters of arity greater than the letters already in the context equation. \qedhere
\end{itemize}
\end{proof}

\subsubsection*{Acknowledgements}
I would like to thank Jan Otop and Manfred Schmidt-Schau\ss{} for introducing me to the topic
and for the question whether recompression generalises to the context unification.

\end{document}